\documentclass[runningheads,table]{llncs}
\usepackage{etoolbox}
\usepackage{paperstyle}
\usepackage{macros}
\def\orcidID#1{\href{http://orcid.org/#1}{\raisebox{-1.25pt}{\includegraphics{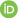}}}}

\newcommand{\PH}[1]{\textcolor{magenta}{PH: #1}}

\newbool{shortversion}
\boolfalse{shortversion}

\begin{document}

\title{Completeness of Synthesis under Realizability Assumptions using Superposition}
%
%
\author{
M\'arton Hajdu\orcidID{0000-0002-8273-2613} \and
Petra Hozzov\'a\orcidID{0000-0003-0845-5811}  \and 
Laura Kov\'acs\orcidID{0000-0002-8299-2714} \and
Eva Maria Wagner\textsuperscript{(\Envelope)}\orcidID{0009-0006-3765-4130}
}
\authorrunning{Hajdu, Hozzov\'a, Kov\'acs, 
Wagner}
\institute{TU Wien, Vienna, Austria\\
eva.maria.wagner@tuwien.ac.at 
} 

\titlerunning{Completeness of Synthesis under Realizability}
\maketitle              
%


\begin{abstract}
 Program synthesis is the task of automatically deriving a program that has been
specified by a user in advance. Combining automated theorem proving with program synthesis enables the automated construction of proven-to-be-correct programs, thereby ensuring software reliability.
In this paper, we consider the superposition-based calculus extended to support synthesis of recursion-free programs allowing reasoning with uncomputable symbols. 
We present cases where the calculus fails and refine it to solve them.
We prove that the refined calculus is sound.
Finally, we also prove completeness in the following sense: if at least one computable program satisfying the given specification exists, we show that the modified calculus finds one.

\end{abstract}
%
%
%

\section{Introduction} \label{section:introduction}
Program synthesis focuses on constructing programs from a given functional specification. 
There are many different approaches for solving different flavors of this problem, including ones that guarantee that the found program satisfies a logical specification~\cite{SynSat,SMTSyn,Z3Syn}; techniques that synthesize a program based on a set of input-output examples~\cite{iosyn1,iosyn2,iosyn3}; and LLM-based methods which require external correctness checking~\cite{llm}.
In this paper, we focus on an automated deductive synthesis approach using a superposition-based theorem prover, in which a saturation-based framework proves a given specification while simultaneously generating code conforming to that specification.

Recent work in~\cite{SynSat} extracts a recursion-free program from a superposition-based proof of a logical specification (requirement) on the program. Our approach explores and revises this framework and solves 
\emph{program synthesis in the presence of uncomputable symbols}~\cite{SynSat,Z3Syn}. Doing so, we impose the following  syntactic restrictions on the program to be synthesized: 
(i) the program should use only so-called \emph{computable} symbols, while (ii) its functional specification may use both computable and {uncomputable} symbols.
Uncomputable symbols include, for example, symbols annotated as such by the user of the synthesis system, allowing for better control of what the output program should (not) use. Uncomputable symbols may, however, also include fresh symbols, such as Skolem functions, introduced during the proving process. 

\begin{motivatingexample*}
\label{ex:motivating}
We motivate our work using an example adapted from~\cite{Reger2018,VampireWS23}. Consider the following constraints from the FLoC 2026 workshop schedule: 
\begin{enumerate}
    \item On Friday ($\fri$), the Vampire ($\vamp$) workshop is taking place. Using the unary predicate $\workshop$ for workshops, we assert  $\workshop(\vamp)$).
    \item On Saturday ($\sat$), the PAAR ($\paar$) workshop is taking place (thus, $\workshop(\paar)$).
    \item Today ($x$) is either Friday ($\fri$) or Saturday ($\sat$).
\end{enumerate}
Our task for program synthesis is to infer what workshop $y$ takes place depending on $x$, with the condition that $\workshop$ is uncomputable, i.e. the answer should not contain $\workshop$.

We encode this instance of program synthesis as follows. 
For readability, we might omit parentheses in unary symbol applications throughout the paper, e.g. we might write $\workshop\vamp$ instead of $\workshop(\vamp)$. 
We are looking for a program that is a witness for $y$ in the following formula:
\begin{equation}
\varphi:\ \ \forall x \exists y.(((x\eqs \fri \lor x \eqs \sat) \land ( x \eqs \fri \! \rightarrow \workshop\vamp) \land (x \eqs \sat \!\rightarrow \workshop\paar))\!\rightarrow \workshop y).\label{eq:motivating-formula}
\end{equation}
A program that is a witness for $y$ in~\eqref{eq:motivating-formula} is called a \emph{solution of the synthesis problem} specified by~\eqref{eq:motivating-formula}; see \Cref{section:problem} for precise formulation of program synthesis and its solution. 
In the case of~\eqref{eq:motivating-formula},  two possible programs found by our approach  are $\IfThenElse{x\eqs\fri}{\vamp}{\paar}$ and $\IfThenElse{x\eqs\sat}{\paar}{\vamp}$. 
\QED
\end{motivatingexample*}

\paragraph{\bf Synthesis and Saturation.} 
In our approach to program synthesis, we build on the saturation-based framework for program synthesis with uncomputable symbols using the superposition calculus, as introduced in~\cite{SynSat}.
While the approach was shown to be correct in~\cite{SynSat}, the question of completeness remained open until now.
In this paper, we close this gap by investigating the \emph{completeness of the calculus under the assumption of realizability}: if a computable program satisfying the specification exists, is the calculus guaranteed to derive it?
We identify properties that make the calculus of~\cite{SynSat} inherently incomplete. We  adjust the calculus accordingly, resulting in our \syncal{} framework (\Cref{section:completeness}). We further prove completeness of \syncal{} under the assumption of realizability (\Cref{section:completeness}).



This paper starts with preliminaries in \Cref{section:preliminaries}, summarizing necessary notions of first-order logic and automated reasoning.
\Cref{section:problem} defines our  synthesis problem and presents a framework for solving it, by revising the setting of~\cite{SynSat}.
Following are \emph{our main contributions}:
\begin{itemize}
    \item We introduce a new superposition-based calculus for synthesis (\Cref{section:calculus}). Our calculus is coined as   Superposition with Realizability Assumptions, in short \syncal.  Motivated by examples that cannot be solved by~\cite{SynSat}, \syncal  
    includes tailored  conditions for term orderings and selection functions.
    \item We prove completeness of our \syncal calculus (\Cref{section:completeness}). By completeness we mean that, 
    if at least one computable program satisfying a given specification exists, our \syncal calculus finds one such program. 
\end{itemize}
We then review related work in \Cref{section:related_work}, and conclude our paper in \Cref{section:summary}. The (omitted) proofs of our results  from \Cref{section:calculus,section:completeness} are given in 
\ifbool{shortversion}{
 the extended version  \cite{synreal-preprint} of this paper.
 }{
\Cref{appendix:calculus,appendix:completeness}, respectively.}

\section{Preliminaries} \label{section:preliminaries}
\paragraph{First-order Logic.} We assume familiarity with standard multi-sorted first-order logic (FOL) with equality, where equality is denoted by $\eqs$. We consider a fixed signature $\Sigma$ consisting of a finite set of \emph{function symbols} with associated arities and a set of \emph{variables} $\mathcal{V}$; variables are not part of the signature. We define the set of \emph{terms} $\mathcal{T}(\Sigma \cup\mathcal{V})$, or simply $\mathcal{T}$ when it is clear from the context,  in the standard way over $\Sigma\cup\mathcal{V}$. 
We denote variables by $x,y,z$; terms by $s, t, l, r, k$; literals by $L, K$;  clauses by $C, D$; and  formulas by $F, G$,  all possibly with indices.
Further, we write $\alpha$ for Skolem constants.
We denote lists of literals by $\ELL$ and $\KAY$. We denote the empty list of literals by $\emptylist$. Let $L$ be a literal and $\ELL$ a list of literals, then $\lapp{L}{\ELL}$ denotes the list of literals having $L$ as a head and $\ELL$ as a tail.
We write $\bar{a}$ for the tuple of variables or terms $a_1,\dots,a_n$.
When a term $t$ is of the form $f(\bar{t})$, we say that $f$ is the \emph{top-level symbol} of $t$.
By $\cnf(F)$ we denote the \emph{clausal normal form} (CNF) of the formula $F$.
We reserve the symbol $\square$ for the \emph{empty clause}
which is logically equivalent to $\bot$.
We write $\hat{L}$ for the literal complementary to $L$.
We write $t\neqs s$ as a shorthand for $\lnot (t\eqs s)$. We use the symbol $\deqs$ to denote either $\eqs$ or $\neqs$. 
An \emph{expression} is a term, literal, clause, or formula. We write $E_1 = E_2$ to denote syntactic equality of two expressions. If an expression does not contain variables, we call it \emph{ground}. 
We write $E[t]$ to denote all (possibly zero) occurrences of the term $t$.
Then, $E[s]$ denotes the expression $E[t]$ where all occurrences of $t$ are replaced by the term $s$. 
Formulas with free variables are considered implicitly universally quantified; that is, we consider closed formulas.
%
A \emph{substitution} $\sigma$ is a mapping from variables to terms such that the set  $\{x\mid \sigma(x)\neq x\}$ of variables is finite. A substitution $\sigma$ is a \emph{unifier} of two expressions $E$ and $E'$ if $E\sigma= E'\sigma$, and is a \emph{most general unifier} (\emph{mgu}) if for every unifier $\theta$ of $E$ and $E'$, there exists substitution $\tau$ such that $\theta=\sigma\tau$. We denote the mgu of expressions $E_1,E_2, E_1',E_2'$ with $\mgu(E_1,E_1')$ and $\mgu((E_1,E_2),(E_1',E_2'))$ for the mgu of tuples. 

\paragraph{Saturation-Based Proving and Superposition.}
A saturation-based prover works with clauses.
To prove a theorem $G$ from axioms $A_1,\dots,A_n$ (assumed to be clauses), the prover: (1) negates and clausifies $G$, obtaining clauses $\cnf(\lnot G)=\{C_1,\dots, C_m\}$; (2) forms the \emph{saturation set} $S = \{C_1,\dots, C_m, A_1,\dots, A_n\}$; (3) repeats the following: chooses clauses $D_1,\dots, D_k \in S$, uses an \emph{inference rule} to derive a new clause $D$ from $D_1,\dots,D_k$, and adds $D$ into $S$. 
Deriving the empty clause $\square$ at any point in step~(3) concludes the proof, because it means that $\lnot G$ is inconsistent with the axioms, which is equivalent to $G$ following from the axioms.
A common \emph{calculus} (a set of inference rules) used in saturation-based provers is the \emph{superposition calculus}~\cite{NieuwenhuisRubio:HandbookAR:paramodulation:2001}.
The calculus is parametrized by a \emph{simplification order} $\succ$ on terms (see next paragraph) and a \emph{selection function}, which selects in each non-empty clause some non-empty subset of literals.
We denote selected literals by underlining them.
An inference rule can be applied to the given premise(s) if the literals selected in the rule are also selected in the premise(s).
For a certain class of selection functions, the superposition calculus is \textit{sound} (if $\square$ is derived from $F$, then $F$ is unsatisfiable) and
\textit{refutationally complete} (if $F$ is unsatisfiable, then $\square$ can be derived from it).

\paragraph{Rewriting and Simplification orders.} \label{subsection:orders}
A binary relation $\to$ over the set of terms is a \emph{rewrite relation} if (i) $l\to r \Rightarrow l\theta\to r\theta$ and (ii) $l\to r \Rightarrow s[l]\to s[r]$ for any terms $l$, $r$, $s$ and substitution $\theta$. We write $\leftarrow$ to denote the inverse of $\to$. We call an ordered pair $l\to r$ a \emph{rewrite rule} if (i) $l$ is not a variable and (ii) $l$ contains all variables that occur in $r$. A \emph{rewrite system} $R$ is a set of rewrite rules. We denote by $\to_R$ the smallest rewrite relation that contains $R$. A term $l$ is \emph{irreducible} in $R$ if there is no $r$ s.t. $l\to_R r \in R$. We denote with $\to^*_R$ the reflexive-transitive closure of $\to_R$. A term $r$ is a \emph{normal form} of a term $l$ w.r.t $R$ if $l\to^*_Rr$ and $r$ is irreducible in $R$. A \emph{rewrite order} is a strict (irreflexive) rewrite relation. A \emph{reduction order} is a well-founded rewrite order. We consider reduction orders which are total on ground terms; such orders are also called \emph{simplification orders}.
A \emph{precedence relation}, denoted by $\gg$, is a total order on the signature $\Sigma$. 

The \emph{lexicographic path order (LPO)}, denoted by $\succ_\lpo$, is parameterized by a precedence relation $\gg$. Let   $s,t$ be terms with  $s=f(s_1,...,s_n)$ and $\succ_{\lpo}^{\lex}$ the standard lexicographic ordering extension of $\succ_{\lpo}$. We write  $ s \succ_\lpo t$ if:
\begin{enumerate}
\item\label{prop:lpo1}  $t$ is a variable and a proper subterm of $s$, or
\item\label{prop:lpo2} there is $i \in \{1, \dots, n\}$ such that $s_i\succeq_\lpo t$, or 
\item\label{prop:lpo3} $t=f(t_1,...,t_n)$, $(s_1,\dots,s_n)\succ_{\lpo}^{\lex} (t_1,\dots,t_n)$ and $s \succ_{\lpo} t_j$, for all $j \in \{1,\dots,n\}$,or 
\item\label{prop:lpo4} $t=g(t_1,...,t_m)$, $f\gg g$ and $s \succ_{\lpo} t_j$, for all $j \in \{1,\dots,m\}$.
\end{enumerate}
It is known that LPOs are simplification orders~\cite{NieuwenhuisRubio:HandbookAR:paramodulation:2001}. 

We view equalities as \emph{bags}, finite multisets, and define and extend the term orderings on bags. The \emph{bag extension} for an ordering $\succ$ on a set $X$ is a binary relation on bags over $X$, denoted by $\succ^{\mathsf{bag}}$, defined as the smallest transitive
relation on bags such that
$\{x,y_1,\ldots,y_n\}\succbag \{x_1,\ldots,x_m,y_1,\ldots,y_n\}$
if $x\succ x_i$ for all $i \in \{1,\dots,m\}$ and $m\geq 0$. If $\succ$ is well-founded, then $\succbag$ is too.
We order equality literals by mapping each equality $s\eqs t$ to the bag $\{s,t\}$ and each disequality $s\neqs t$ to the bag $\{s,s,t,t\}$, and then using $\succ^{\mathsf{bag}}$ over these bags. Finally, we order clauses by taking the bag extension of the bag extension for literals. We only write $E_1 \succ E_2$ for two expressions $E_1$ and $E_2$ when it is clear from the context which ordering is used. We also write $E_1 \succeq E_2$ instead of $E_1 \succ E_2 \lor E_1 = E_2$.

\section{The Synthesis Problem and How to Solve It} \label{section:problem}
In this section, we adjust  notions~\cite{SynSat}  to the purposes of our work. 
%
We assume a fixed signature $\Sigma$.

\begin{definition}[Synthesis specification]
Let $\Sigma_c$ be a subset of $\Sigma$ and $\varphi$ a closed formula in first-order logic over $\Sigma$ of the form:
\begin{equation}
\forall \bar{x}\exists y. F[\bar{x}, y]. \label{eq:spec}
\end{equation}
We define a \emph{synthesis specification}, or simply \emph{specification},  $\Lambda$ as the pair $\langle \Sigma_c, \varphi\rangle$. We call any symbol in $\Sigma_c$ \emph{computable w.r.t. $\Lambda$}, or just \emph{computable} when $\Lambda$ is clear from the context. Further, 
any symbol in $\Sigma\setminus\Sigma_c$  is \emph{uncomputable w.r.t. $\Lambda$}, or  just \emph{uncomputable}. We denote the set of uncomputable symbols by $\Sigma_u$.

An expression containing an uncomputable symbol is called \emph{uncomputable}.
Expressions that are not uncomputable are called \emph{computable}.\QED
\end{definition}%
%
To synthesize programs for specifications, we use an if-then-else term constructor, denoted by $\iteF$.
\begin{definition}[$\iteF$ and program terms]\label{def:programterms}
We define the \emph{conditional term constructor} $\iteF$ as follows.
Let $F$ be a formula over $\Sigma$ and $p, q$ be terms.
When $F$ is true, then $\iteF(F, p, q)$ is interpreted as the interpretation of $p$; otherwise, $\iteF(F,p,q)$ is interpreted as the interpretation of $q$.
We call $F$ the \emph{$\iteF$-condition}. 

We define \emph{program terms} as the smallest set $\mathcal{P}$ such that (i) $\mathcal{T}\subseteq\mathcal{P}$, and (ii) $\iteF(l\eqs r,p,q)\in\mathcal{P}$ for any $l,r\in\mathcal{T}$, and $p,q\in\mathcal{P}$. We denote program terms by $p$, $q$, possibly with indices. A term that does not contain a conditional term constructor is called a \emph{simple term}.\QED
\end{definition}
%
We note that any term with the $\iteF$ symbol containing arbitrary Boolean expressions as conditions can be translated into an equivalent program term in $\mathcal{P}$.

\begin{definition}[Solution for a specification]
Let $\Lambda=\langle \Sigma_c, \forall\bar{x}\exists y.F[\bar{x},y]]\rangle$ be a specification. A computable $p[\bar{x}]\in\mathcal{P}$ is called a \emph{solution} for the specification $\Lambda$ if $\forall\bar{x}.F[\bar{x},p[\bar{x}]\theta]$ is valid for an arbitrary substitution $\theta$ grounding $p[\bar{a}]$ for fresh constants $\bar{a}$.\QED
\end{definition}
Note that if $\bar{x}$ are the only variables in $p[\bar{x}]$, then the condition for $p[\bar{x}]$ being a solution reduces to $\forall\bar{x}.F[\bar{x},p[\bar{x}]]$ being valid.
On the other hand, if $p[\bar{x}]$ contains variables $\bar{z}$ other than $\bar{x}$, then $p[\bar{x}]$ represents a class of solutions, one for each grounding of $\bar{z}$.
\begin{definition}[Realizable specification]
Let $\Lambda$ be a specification. We call $\Lambda$ \emph{realizable} if there exists a computable program term that is a solution to it.\QED
\end{definition}
We are now ready to state the problem that is the focus of the present paper.
\begin{mdframed}[frametitle={\fcolorbox{black}{white}{~The Synthesis Problem under Realizability Assumptions~}}, innertopmargin=0pt, innerbottommargin=8pt, frametitleaboveskip=-5pt, frametitlealignment=\center, frametitlefont=\scshape]
Given a realizable specification $\Lambda$, the \emph{synthesis problem under realizability assumptions} is the task of finding a solution to $\Lambda$.
\end{mdframed}
In the following, we refer to the synthesis problem under realizability assumptions simply as \emph{the synthesis problem}.


\paragraph{Saturation-based synthesis framework.}
One approach to solving synthesis problems is the saturation-based synthesis framework introduced in~\cite{SynSat}, which extends a saturation-based first-order theorem prover into a synthesizer.
The framework constructs a program for specification $\Lambda$ in parallel with searching for a proof that specification~\eqref{eq:spec} holds.
Intuitively, $y$ in~\eqref{eq:spec} represents the output for the inputs $\bar{x}$, and thus substitutions into $y$ in the proof correspond to fragments of the sought program.
To track these substitutions throughout the proof, the framework of~\cite{SynSat} uses a mechanism called answer literals~\cite{Green69}. 
In this paper, we use constrained clauses, called \emph{answer clauses}, rather than adding answer literals to clauses.
\begin{definition}[Answer clause]
Let $C$ be a clause and $p$ a program term. We call the expression $\ansC{C}{p}$ an \emph{answer clause}. We call $p$ the \emph{answer} for the answer clause $\ansC{C}{p}$.
We denote answer clauses with $\AC$ and $\AD$, possibly with indices. Given a substitution $\sigma$, we write $\ansC{C}{p}\sigma$ to denote $\ansC{C\sigma}{p\sigma}$.\QED 
\end{definition}
W.l.o.g. we assume that axioms $A_1,\dots,A_n$ are a part of specification~\eqref{eq:spec}; that is,  $F[\bar{x}, y]$ is of the form $(A_1\land\dots\land A_n)\rightarrow G[\bar{x}, y]$.
To derive a program, the framework first preprocesses~\eqref{eq:spec} like a saturation-based prover would:
the formula $\lnot\forall\bar{x}\exists y.F[\bar{x}, y]$ is converted to the equivalent $\exists\bar{x}\forall y.\lnot F[\bar{x}, y]$ and skolemized, obtaining the equisatisfiable $\forall y.\lnot F[\bar{\alpha}, y]$,\footnote{Since the skolems $\bar{\alpha}$ represent the input of the sought program, they are computable.} which is then clausified, resulting into $\cnf(\lnot F[\bar{\alpha}, y]) = \{ C_1,\dots,C_m\}$.  
The framework then extends the clauses $C_1,\dots,C_m$, among which are also the axioms, into answer clauses $\ansC{C_1}{y},\dots,\ansC{C_m}{y}$,\footnote{In~\cite{SynSat}, only the clauses $C$ containing $y$ are extended into answer clauses $\ansC{C}{y}$. In this paper we extend \emph{all} clauses $C$ in preprocessing, including axioms, into answer clauses $\ansC{C}{y}$.}
and initializes the saturation set $S = \{\ansC{C_1}{y},\dots,\ansC{C_m}{y}\}$, called \emph{the initial set from $\Lambda$}.
We denote the consecutive steps of clausification and extending clauses into answer clauses by $\ansC{\cnf(\lnot F[\bar{\alpha}, y])}{y}$.
\begin{definition}[Semantics of answer clauses]
\label{def:semantics}
$\ansC{C}{p}$ \emph{is true} with respect to a given specification $\Lambda=\langle \Sigma_c, \forall\bar{x}\exists y.F[\bar{x},y]\rangle$, if the universal closure of $C\lor F[\bar{\alpha},p]$ is true, where w.l.o.g. we assume $\ansC{C}{p}$ and $F[\bar{\alpha},y]$ to have disjoint sets of variables.\footnote{\cite{SynSat} uses analogous semantics while denoting $\ansC{C}{p}$ by $C\lor\ans(p)$.}\QED 
\end{definition}
To work with answer clauses, we use an  extension of the superposition calculus~\cite{NieuwenhuisRubio:HandbookAR:paramodulation:2001}.
Using this calculus, we \emph{saturate} $S$ like a saturation-based prover: we extend the set  $S$ by new clauses derived by rules from the calculus based on premises already present in the set.
We call $S$ the \emph{set saturated from $\Lambda$}.
The work~\cite{SynSat} proves that the saturation approach is sound: it only derives valid answer clauses, and thus when the answer clause $\ansC{\square}{p[\bar{\alpha}]}$ is derived, it is guaranteed that $p[\bar{x}]$ is a solution for $\Lambda$.
However, \cite{SynSat} makes no completeness claims.
In the following section, we show an example for which the approach fails to derive a program even though a computable program exists.

\section{Superposition with Synthesis} \label{section:calculus}

In this section, we present our new calculus \syncal,  adapted from~\cite{SynSat}. We discuss the orderings and selection constraints that guide \syncal 
(\Cref{subsec:calculus}).
We also introduce an abstraction mechanism that separates computable and uncomputable terms via an inference rule, in addition to \syncal (\Cref{subsec:abstract}).

\subsection{The \syncal Calculus}\label{subsec:calculus}
Our new calculus for  Superposition with Realizability Assumptions, in short \syncal, is summarized in \Cref{fig:syncal}. 
Compared to~\cite{SynSat}, we do not use so-called computable unifiers (abstracting unifiers preventing uncomputable symbols from appearing in the answer term), but most general unifiers and only allow an inference when the answer term after applying $\sigma$ is computable.
Further, we include the rule $\Sup_U$, which unifies the answer terms $p, q$ from premises, while~\cite{SynSat} included a superposition rule which did not unify $p, q$ but rather added a constraint $p\neqs q$ into the resulting clause. 
Finally, for simplicity, we work only with equality predicates;  thus, we do not include binary resolution or factoring rules as in~\cite{SynSat}.
\ifbool{shortversion}{}{
\begin{remark}\label{remark:rules}
    The two variants of Superposition are both needed in the synthesis setting. Applying the superposition rule branches the proof upon the equality literal. The $\Sup_C$ rule reflects this branching in the answer by introducing an $\iteF$ constructor with the equality literal in the condition. However, since the answers cannot contain uncomputable symbols, and even a computable non-ground condition could become uncomputable by a subsequent substitution, we need an additional mechanism for applying superposition. This is provided by $\Sup_U$, which is applied if the answers from the premises are unifiable. Since it is not clear upfront which rule makes progress, our framework uses both of them in parallel.
\end{remark}
}

We next establish soundness of  \syncal; this result guarantees that the programs  derived by \syncal are solutions of the synthesis specification.\footnote{We provide detailed proofs for all our results in \ifbool{shortversion}{the extended version \cite{synreal-preprint} of this paper.}{the Appendix.}} 
\begin{restatable}[Soundness]{theorem}{Soundness}\label{th:soundness}
The \syncal calculus is sound with respect to the semantics of answer clauses.
\end{restatable}
\noindent
It follows from~\Cref{th:soundness} that \syncal derives solutions to synthesis specifications.
\begin{restatable}[Solution to the Synthesis Problem]{corollary}{Cor}\label{cor:soundness}
Given a synthesis specification $\Lambda=\langle \Sigma_c, \forall\bar{x}\exists y.F[\bar{x},y]\rangle$, if \syncal derives $\ansC{\square}{p[\bar{\alpha}]}$ from the initial set of $\Lambda$, then $p[\bar{x}]$ is a solution for $\Lambda$.
\end{restatable}

\begin{figure}[t]
\begin{center}
\begin{tabular}{c c l}

\multirow{4}{*}{
\AxiomC{$\sel{l\eqs r}\lor \ansC{C}{p}$}
\AxiomC{$\ansC{\sel{s[l']\deqs t}\lor D}{q}$}
\LeftLabel{($\Sup_C$)}
\BinaryInfC{$\ansC{s[r]\deqs t\lor C \lor D}{\iteF(l\eqs r,q,p)}\sigma$}
\DisplayProof
}
& \multirow{4}{*}{where} & (1) $\sigma=\mgu(l,l')$\\
& & (2) $l'$ is not a variable\\
& & (3) $r\sigma \nsucceq l\sigma$, $t\sigma  \nsucceq s[l']\sigma$ \\
& & (4) $\mathsf{ite}(l\eqs r,q,p)\sigma$ is computable \\

\\

\multirow{4}{*}{
\AxiomC{$\ansC{\sel{l\eqs r}\lor C}{p}$}
\AxiomC{$\ansC{\sel{s[l']\deqs t}\lor D}{q}$}
\LeftLabel{($\Sup_U$)}
\BinaryInfC{$\ansC{s[r]\deqs t\lor C \lor D}{p}\sigma$}
\DisplayProof
}
& \multirow{4}{*}{where} & (1) $\sigma=\mgu((l,p),(l',q))$\\
& & (2) $l'$ is not a variable\\
& & (3) $r\sigma \nsucceq l\sigma$, $t\sigma  \nsucceq s[l']\sigma$ \\
& & (4) $p\sigma$ is a computable simple term\\

\\

\multirow{2}{*}{
\AxiomC{$\ansC{\sel{s\neqs t}\lor C}{p}$}
\LeftLabel{($\EqRes$)}
\UnaryInfC{$\ansC{C}{p}\sigma$}
\DisplayProof
}
& \multirow{2}{*}{where} & (1) $\sigma=\mgu(s,t)$\\
& & (2) $p\sigma$ is computable\\

\\[.5em]

\multirow{3}{*}{
\AxiomC{$\ansC{\sel{s\eqs t}\lor l\eqs r \lor C}{p}$} 
\LeftLabel{($\EqFac$)}
\UnaryInfC{$\ansC{s\eqs t\lor t\neqs r \lor C}{p}\sigma$}
\DisplayProof
}
& \multirow{3}{*}{where} & (1) $\sigma=\mgu(s,l)$\\
& & (2) $t\sigma \nsucceq s\sigma$, $r\sigma \nsucc t\sigma$\\
& & (3) $p\sigma$ is computable\\[.5em]
\end{tabular}
\end{center}
\caption{Rules of \syncal, where all maximal literals in a clause are selected (denoted underlined), and $\ansC{C}{p}\sigma$ denotes the application of $\sigma$ on $\ansC{C}{p}$.}
\label{fig:syncal}
\end{figure}
\noindent
The rules of \syncal are, however, not the only ingredient necessary to guarantee that a solution for a realizable specification will be derived.
The following example illustrates how a derivation can get stuck.

\begin{example}\label{ex:workshop}
Our motivating example corresponding to the specification~\eqref{eq:motivating-formula} can be expressed as a synthesis specification with computable symbols $\fri$, $\sat$, $\paar$ and $\vamp$. That is, 
\(\Lambda : \ \langle\{\fri,\sat,\paar,\vamp\},\varphi\rangle\).
We consider predicate symbols as function symbols mapping values to $\{\top,\:\perp\}$ in a two-sorted logic.
We will try to derive a solution for $\Lambda$ using \syncal.
The initial set from $\Lambda$, where $\alpha$ is the Skolem constant introduced for $x$ and numbers denote clause labels, is
$\{\text{(1) }\ansC{\alpha\eqs \fri \lor \sel{\alpha\eqs \sat}}{y},
\text{(2) }\ansC{\sel{ \alpha\neqs \sat} \lor \workshop\paar}{y},
\text{(3) }\ansC{\sel{\alpha\neqs \fri} \lor \workshop\vamp}{y},
\text{(4) }\ansC{\lnot\sel{\workshop y}}{y}\}$.
We use a selection function that selects either (i) all maximal literals or (ii) at least one negative literal.\footnote{Such a selection function is called \emph{well-behaved} and is sufficient to achieve refutational completeness in standard superposition reasoning.}
Let $\succ$ be an LPO with the following symbol precedence: $\sat\gg\fri\gg\paar\gg\vamp\gg\workshop\gg\alpha$.
We obtain the following derivation, where we write $\SE$ for a $\Sup$ inference followed by an $\EqRes$ inference:
\begin{center}
\small
\begin{prooftree}
\AxiomC{\scriptsize (1)}
\noLine
\UnaryInfC{$\ansC{\alpha\eqs \fri \lor \sel{\alpha\eqs \sat}}{y}$}
\AxiomC{\scriptsize (2)}
\noLine
\UnaryInfC{$\ansC{\sel{ \alpha\neqs \sat} \lor \workshop\paar}{y}$}
\LeftLabel{\scriptsize ($\SE$)}
\BinaryInfC{$\ansC{\sel{\alpha\eqs \fri} \lor \workshop\paar}{y}$}
\AxiomC{\scriptsize (3)}
\noLine
\UnaryInfC{$\ansC{\sel{\alpha\neqs \fri} \lor \workshop\vamp}{y}$}
\LeftLabel{\scriptsize ($\SE$)}
\BinaryInfC{$\ansC{\sel{\workshop\paar}\lor \workshop\vamp}{y}$}
\AxiomC{\scriptsize (4)}
\noLine
\UnaryInfC{$\ansC{\lnot\sel{\workshop y}}{y}$}
\LeftLabel{\scriptsize ($\SE$)}
\BinaryInfC{$\ansC{\sel{\workshop\vamp}}{\paar}$}
\AxiomC{\scriptsize (4)}
\noLine
\UnaryInfC{$\ansC{\lnot\sel{\workshop y}}{y}$}
\BinaryInfC{{\scriptsize\color{gray}{}(inference not possible)}}
\noLine
\UnaryInfC{}
\end{prooftree}
\end{center}
Here we get stuck, as  no more inferences are possible. Superposition into $\text{(4) }\ansC{\lnot \workshop y}{y}$ with $\ansC{\workshop\vamp}{\paar}$ cannot be applied, because $\vamp$ and $\paar$ do not unify; moreover, we cannot create an $\iteF$ with condition $\workshop\vamp$, because $\workshop\in\Sigma_u$.
However, as mentioned before, $\Lambda$ is realizable, e.g. by the solution $\iteF(x\eqs\fri,\vamp,\paar)$.
We did not derive a solution due 
to the selection of computable literals in clauses which also contained uncomputable literals.\QED

\end{example}
\begin{remark}\label{remark:synsat}
Since~\cite{SynSat} does not restrict selection nor simplification orders in any way, the failed derivation from \Cref{ex:workshop}  can be  replicated by the calculus of~\cite{SynSat}.
Therefore, the calculus~\cite{SynSat} does not always find a solution if one exists -- it is not complete under realizability assumptions.\QED
\end{remark}
Based on~\Cref{ex:workshop}, we observe that we have to select uncomputable literals before computable ones.
To this end, we make uncomputable terms bigger than computable terms in the simplification order and restrict the selection functions admissible for \syncal to always select all maximal literals. 
We formally refine the simplification order used by \syncal as follows.
\begin{definition}[Partitioned ordering]
An ordering $\succ$ on terms is called \emph{partitioned} if for any ground uncomputable term $s$ and any ground computable term $t$ it holds that $s\succ t$.\QED
\end{definition}
We require that the simplification order that parameterizes \syncal is partitioned.
In the rest of the paper, we will use an LPO with a precedence function where any uncomputable symbol is greater than any computable symbol. Not only can these types of LPOs be used for \syncal, but also specific KBOs that are partitioned orderings, \ifbool{shortversion}{see~\cite{synreal-preprint} for further discussion.}{see~\Cref{appendix:ordering} for further discussion.}
\begin{example}\label{ex:workshop2}
Let $\succ$ be an LPO based on $\workshop \gg \vamp \gg \paar \gg \sat \gg \fri$.
With a selection function selecting all maximal literals w.r.t. $\succ$, we find a solution for $\Lambda$ from~\Cref{ex:workshop}:

\begin{center}
\begin{prooftree}
\def\defaultHypSeparation{\hskip .05in}
\AxiomC{\scriptsize (2)}
\noLine
\UnaryInfC{$\ansC{\alpha\neqs \sat \lor \sel{\workshop\paar}}{y}$}
\AxiomC{\scriptsize (4)}
\noLine
\UnaryInfC{$\ansC{\lnot\sel{\workshop y}}{y}$}
\LeftLabel{\scriptsize ($\SE$)}
\BinaryInfC{$\ansC{\sel{\alpha\neqs \sat}}{\paar}$}
\AxiomC{\scriptsize (1)}
\noLine
\UnaryInfC{$\ansC{\alpha\eqs \fri \lor \sel{\alpha\eqs \sat}}{y}$}
\LeftLabel{\scriptsize ($\SE$)}
\BinaryInfC{$\ansC{\sel{\alpha\eqs \fri}}{\paar}$}
\AxiomC{\scriptsize (4)}
\noLine
\UnaryInfC{$\ansC{\lnot\sel{\workshop y}}{y}$}
\AxiomC{\scriptsize (3)}
\noLine
\UnaryInfC{$\ansC{\alpha\neqs \fri \lor \sel{\workshop\vamp}}{y}$}
\RightLabel{\scriptsize ($\SE$)}
\BinaryInfC{$\ansC{\sel{\alpha\neqs \fri}}{\vamp}$}
\LeftLabel{\scriptsize ($\SE$)}
\BinaryInfC{$\ansC{\square}{\itecons(\alpha\eqs\fri,\vamp,\paar)}$}
\end{prooftree}
\end{center}
\vspace{-1em}
\QED
\end{example}
%
%
\subsection{Abstraction of Computable Terms in Uncomputable Literals}\label{subsec:abstract}
While our motivating example~\eqref{eq:motivating-formula} shows that inferences with uncomputable literals should be applied before inferences with computable literals, sometimes this cannot be enforced, as some uncomputable equations can only be used in a superposition inference after performing a superposition into them with a computable equation.
\begin{example}\label{ex:abs}
Consider constants $a,b,c,d,e$, unary symbols $f, g, h$, and the specification:
\[\langle \{a,b,c,d,e\},\ \exists y.(fd\neqs e \lor (d\neqs c\land gy\eqs ga) \lor (fc\eqs e\land hy\eqs hb))\rangle.\]
We perform the following derivation using an LPO $\succ$ with $h\gg g\gg f\gg e\gg d\gg c\gg b\gg a$. 
The leaves are obtained from the initial set of the above specification:
\begin{center}
\begin{prooftree}
\def\defaultHypSeparation{\hskip .1in}
\AxiomC{\scriptsize (1)}
\noLine
\UnaryInfC{$\ansC{\sel{fd\eqs e}}{y}$}
\AxiomC{\scriptsize (2)}
\noLine
\UnaryInfC{$\ansC{d\eqs c\lor \sel{gy\neqs ga}}{y}$}
\LeftLabel{{\scriptsize ($\ER$)}}
\UnaryInfC{$\ansC{\sel{d\eqs c}}{a}$}
\LeftLabel{{\scriptsize ($\Sup_U$)}}
\BinaryInfC{$\ansC{\sel{fc\eqs e}}{\ite{d\eqs c}{y}{a}}$}
\AxiomC{\scriptsize (3)}
\noLine
\UnaryInfC{$\ansC{fc\neqs e\lor \sel{hy\neqs hb}}{y}$}
\LeftLabel{{\scriptsize ($\ER$)}}
\UnaryInfC{$\ansC{\sel{fc\neqs e}}{b}$}
\LeftLabel{{\scriptsize ($\Sup_C$)}}
\BinaryInfC{{\scriptsize\color{gray}{}(inference not possible)}}
\end{prooftree}
\end{center}
The last step of the derivation is not possible, because the answers $b$ and $\ite{d\eqs c}{y}{a}$ do not unify, and we cannot construct an $\itecons$ in the answer, because the condition $fc\eqs e$ is uncomputable.\QED
\end{example}
%
The general observation we make from \Cref{ex:abs} is that we should resolve all uncomputable literals, even those containing computable subterms, before applying inferences with computable symbols.
For this purpose, we \emph{abstract the computable subterms of uncomputable literals}, thus separating reasoning with uncomputable and computable symbols.
We do this by adding an auxiliary inference rule $\Abs$, which replaces its premise by its consequence, indicated by the crossed-out premise.
We \emph{apply $\Abs$ exhaustively on each clause before we allow any further inferences}:

\begin{center}
\begin{tabular}{c c c l}

\multirow{7}{*}{
\AxiomC{\cancel{$\ansC{\sel{s[k]\deqs t}\lor C}{p}$}}
\LeftLabel{($\Abs$)}
\UnaryInfC{$\ansC{s[x]\deqs t\lor x\neqs k\lor C}{p}$}
\DisplayProof
}
& \multirow{7}{*}{where} & (1)& $s\not\preceq t$ \\
& & (2)& there is a substitution $\theta$ s.t. $s[k]\theta$ \\
& & & is uncomputable and $k\theta$ is computable\\
& & (3)& no proper superterm of $k$ in $s[k]$ \\
& & & satisfies property (2) \\
& & (4)& $k$ is not a variable \\
& & (5)& $x$ is a fresh variable \\
\end{tabular}
\end{center}
\begin{restatable}{lemma}{absSound}\label{lemma:absSound}
    The rule $\Abs$ is sound with respect to the semantics of answer clauses.
\end{restatable}
To apply the rule in practice, we use the following syntactic conditions in place of condition (2) of $\Abs$.
\begin{restatable}{lemma}{absConditions}\label{lemma:abstraction-conditions}
Condition (2) of rule $\Abs$ holds iff (2a) $k$ is computable, and (2b) either $s[k]$ is uncomputable, or it contains a variable not contained in $k$, and there exists at least one uncomputable term.
\end{restatable}

\begin{example}\label{ex:abs-continued}
The derivation from~\Cref{ex:abs} changes as follows when using $\Abs$:
%
%
\begin{center}
\smaller
\begin{prooftree}
\def\defaultHypSeparation{\hskip .1in}
\AxiomC{\scriptsize (2)}
\noLine
\UnaryInfC{\cancel{$\ansC{d\eqs c\lor \sel{gy\neqs ga}}{y}$}}
\LeftLabel{{\scriptsize ($\Abs$)}}
\UnaryInfC{$\ansC{d\eqs c\lor \sel{gy\neqs gx}\lor x\neqs a}{y}$}
\LeftLabel{{\scriptsize ($\ER$)}}
\UnaryInfC{$\ansC{\sel{d\eqs c}\lor \sel{x\neqs a}}{x}$}
\LeftLabel{{\scriptsize ($\ER$)}}
\UnaryInfC{$\ansC{\sel{d\eqs c}}{a}$}
\AxiomC{\scriptsize (1)}
\noLine
\UnaryInfC{\cancel{$\ansC{\sel{fd\eqs e}}{y}$}}
\LeftLabel{{\scriptsize ($\Abs$)}}
\UnaryInfC{$\ansC{\sel{fx\eqs e}\lor x\neqs d}{y}$}
\AxiomC{\scriptsize (3)}
\noLine
\UnaryInfC{\cancel{$\ansC{fc\neqs e\lor \sel{hy\neqs hb}}{y}$}}
\LeftLabel{{\scriptsize ($\Abs$)}}
\UnaryInfC{$\ansC{fc\neqs e\lor \sel{hy\neqs hx}\lor x\neqs b}{y}$}
\LeftLabel{{\scriptsize ($\ER$)}}
\UnaryInfC{\cancel{$\ansC{\sel{fc\neqs e}\lor \sel{x\neqs b}}{x}$}}
\LeftLabel{{\scriptsize ($\Abs$)}}
\UnaryInfC{$\ansC{\sel{fz\neqs e}\lor z\neqs c\lor \sel{x\neqs b}}{x}$}
\LeftLabel{{\scriptsize ($\ER$)}}
\UnaryInfC{$\ansC{\sel{fz\neqs e}\lor z\neqs c}{b}$}
\LeftLabel{{\scriptsize ($\SE$)}}
\BinaryInfC{$\ansC{\sel{y\neqs d} \lor y\neqs c}{b}$}
\LeftLabel{{\scriptsize ($\ER$)}}
\UnaryInfC{$\ansC{\sel{d\neqs c}}{b}$}
\LeftLabel{{\scriptsize ($\SE$)}}
\BinaryInfC{$\ansC{\square}{\ite{d\eqs c}{b}{a}}$}
\end{prooftree}
\end{center} 
Our new $\Abs$ rule allowed us to solve the specification of \Cref{ex:abs}, deriving the program $\ite{d\eqs c}{b}{a}$.\\
\QED
\end{example}



\section{Completeness of\, \syncal under Realizability Assumptions}
\label{section:completeness}
In this section, we prove that the \syncal calculus is \emph{complete with respect to realizability assumptions}.
In particular, we show that if there exists a solution to the specification $\Lambda$, then \syncal used exhaustively with the rule $\Abs$ is guaranteed to derive some (possibly different) computable program, which is also a solution for $\Lambda$ since \syncal is sound (see \Cref{subsec:calculus}).
Our main result is the following theorem.

\begin{theorem}[Completeness under Realizability] \label{theorem:completeness}
Let $\Lambda=\langle \Sigma_c, \forall\bar{x}\exists y.F[\bar{x},y]\rangle$ be a specification. If $\Lambda$ is realizable, then \syncal with $\Abs$ derives an answer clause of the form $\ansC{\square}{p[\bar{\alpha}]}$ from $\Lambda$ where $p[\bar{x}]$ is a computable program term that is a solution to $\Lambda$.
\end{theorem}%
We prove Theorem~\ref{theorem:completeness} similarly to the completeness proof of the superposition calculus~\cite{10.1007/3-540-55253-7_22}, by contradiction.
We assume that $S$ is a set saturated from $\Lambda$ up to redundancy and abstracted (\Cref{def:saturation}), and that $\ansC{\square}{p[\bar{\alpha}]}\notin S$ for any $p[\bar{\alpha}]$. Yet, we assume that $\Lambda$ is realizable and that $t[\bar{x}]$ is a solution to it -- i.e., that $\forall \bar{x}.F[\bar{x}, t[\bar{x}]]$ is valid.
W.l.o.g. we assume $t[\bar{x}]$ does not contain any variables except for $\bar{x}$.
We then consider a grounding of $S$ using $t[\bar{\alpha}]$ (\Cref{def:uncomp-grounding}) and construct an interpretation for it (\Cref{def:superposition-model}).
However, since the grounding of $S$ is equisatisfiable with $\neg F[\bar{\alpha}, t[\bar{\alpha}]]$, which is unsatisfiable because $\forall \bar{x}.F[\bar{x}, t[\bar{x}]]$ is valid, the interpretation should not be a model of $S$.
We then consider the smallest clause $\CC$ in the grounding that is false in the interpretation.
Based on it, we either find an even smaller clause in the grounding that is also false in the interpretation, or show that the clause in $S$, instance of which is $\CC$, should have been removed by $\Abs$, since $S$ is abstracted. In either case, we obtain a contradiction. 
We split this part of the proof into two steps: first, in~\Cref{lemma:computable-clauses-completeness} we show that all computable clauses in the grounding of $S$ are satisfied in the model;  then in~\Cref{lemma:grounding-satisfied} we extend this property to all uncomputable clauses.

Our construction involves substituting program terms into clauses -- the term $t[\bar{\alpha}]$ is substituted for $y$ in $\ansC{\cnf(\lnot F[\bar{\alpha},y])}{y}$.
If $t[\bar{\alpha}]$ contains $\iteF$ terms, we have to unroll them by creating multiple instances of the clause and adding the condition(s) from the $\iteF$ into them. We collect the unrolled condition literals in lists and separate them from the original clause, using so-called conditional clauses, as defined below.
\begin{definition}[Conditional clause, c-clause]
Let $C$ be a clause and $\ELL$ a list of literals.
A \emph{conditional clause}, \emph{c-clause} for short,  is an expression of the form $\cC{C}{\ELL}$. We call $\ELL$ the \emph{condition} of the c-clause $\cC{C}{\ELL}$. The c-clause $\cC{C}{\ELL}$ is logically equivalent to the clause $C\lor \bigvee_{L\in\ELL}L$. We  denote c-clauses with $\CC$ and $\CD$. We use the notation $\cCs{C}$ to denote the c-clause $\cC{C}{\emptylist}$.\QED 
\end{definition}
\noindent
We extend the ordering $\succ$ over clauses to c-clauses as follows.

\begin{definition}[Ordering on c-clauses]
We define an \emph{ordering $\succ$ over c-clauses} as follows. First, we extend $\succ$ to lists of literals. We have $\ELL\succ\emptylist$ for any non-empty $\ELL$ and $\lapp{L}{\ELL}\succ \lapp{K}{\KAY}$ if either (i) $L\succ K$ or (ii) $L=K$ and $\ELL\succ \KAY$. We have $\cC{C}{\ELL}\succ \cC{D}{\KAY}$ if either (i) $C\succ D$, or (ii) $C=D$ and $\ELL\succ \KAY$.\footnote{This ordering preserves well-foundedness.}\QED
\end{definition}
We now define a model for sets of c-clauses. As usual in superposition completeness proofs, we use rewrite systems as interpretations, where any equation $s\eqs t$ is true in a rewrite interpretation $R$, denoted $R\models s\eqs t$, if $s$ and $t$ have the same normal form in $R$.
As usual, when $s, t$ are ground, $R\models s\not\eqs t \iff R\nmodels s\eqs t$.
A c-clause $\cC{C}{\ELL}$ is true in $R$ if $C$ or $\ELL$ contain at least one literal $L$ such that $R\models L$. 

\begin{definition}[Model for c-clauses]\label{def:superposition-model}
Let $S$ be a set of ground c-clauses. For every c-clause $\CC \in S$, we define in parallel two sets of term rewrite rules $R^{\CC}_S$ and $R^{\prec \CC}_S$ as \emph{partial interpretations} by induction on the relation $\succ$ on ground c-clauses. First, we define:
  \[
    R^{\prec \CC}_S := \bigcup_{\CD \prec \CC, \CD\in S} R^{\CD}_S.
  \]
A ground c-clause $\CC$ of the form $\cC{\underline{l\eqs r} \lor C}{\ELL}$ is called \emph{productive} if
\begin{enumerate}
\item $\ELL$ is computable,\label{productive:comp}
\item $l\eqs r \lor C$ is false in $R^{\prec \CC}_S$,\label{productive:false1}
\item if $l\eqs r \lor C$ is not computable, then all literals in $\ELL$ are false in $R^{\prec \CC}_S$,\label{productive:false2}
\item $l\eqs r$ is strictly maximal in $\CC$,\label{productive:maximal}
\item $l\succ r$,\label{productive:ordered}
\item $C$ is false in $R^{\prec \CC}_S\cup\{l\to r\}$,\label{productive:false3}
\item $l$ is irreducible in $R^{\prec \CC}_S$.\label{productive:irreducible}
\end{enumerate}
In this case, we also say that $\CC$ \emph{produces} the rule $l\to r$. Now we define 

\[
  \begin{array}{rcl}
  R^{\CC}_S & := &
    \left\{
      \begin{array}{ll}
          R^{\prec \CC}_S \cup \{ l \to r \}, & \text{ if $\CC$ produces $l\to r$;} \\
          R^{\prec \CC}_S, & \text{ otherwise.}
      \end{array} 
    \right.
  \end{array}
\]
Finally, we define the \emph{total interpretation $R_S$ for $S$} as $\bigcup_{\CC \in S} R^{\CC}_S$. As usual, an interpretation which satisfies a set of c-clauses $S$ is called a model of $S$.\QED
\end{definition}
%
%
We state two standard properties about the model.
\begin{restatable}{lemma}{Rproperties}\label{lem:Rproperties}
Let $S$ be a set of ground c-clauses.
\begin{enumerate}
    \item $R_S$ is a convergent rewrite system.
    \item $R^{\CC}_S \models \CC$ if and only if for all $\CD \succ \CC$ we have $R^{\CD}_S \models \CC$, if and only if $R_S \models \CC$.
\end{enumerate}
\end{restatable}
\noindent
We will construct a model for a set which is \emph{saturated up to redundancy} and \emph{abstracted}, with the following redundancy notions.
\begin{definition}[Redundant answer clause/inference]
\label{def:redundant-clause-inference}
An answer clause $\AC$ is \emph{redundant w.r.t.\ $S$} if every ground instance of $\AC$ follows from smaller ground instances in $S$.
Let $N$ be the set of all ground instances of answer clauses in $S$. An inference $\AC_1,...,\AC_n\vdash \AD$ is \emph{redundant w.r.t.\ $S$} if for each $\theta$ grounding for $\AC_1, \ldots, \AC_n$ and $\AD$ either
\begin{enumerate}
\item $\AD\theta\succ \AC_i\theta$ for some $1\le i\le n$, or
\item $\AD\theta$ follows from the set $\{C \mid C\in N\text{ and }\AC_i\theta\succ C \text{ for some }1\le i\le n\}$.\QED
\end{enumerate}
\end{definition}
\begin{definition}[Saturation up to redundancy, abstracted set]\label{def:saturation}
A set of answer clauses $S$ is \emph{saturated up to redundancy} if, given non-redundant answer 
clauses $\AC_1,...,\AC_n\in S$, any \syncal inference $\AC_1,..., \AC_n \vdash \AD$ is redundant w.r.t.\ $S$.\\
If there is no clause $\AC\in S$ on which $\Abs$ would apply, then we call $S$ \emph{abstracted}.\QED
\end{definition}
In the remainder of this section, we assume $\Lambda=\langle\Sigma_c,\forall \bar{x}.\exists y.F[\bar{x},y]\rangle$ to be a fixed but arbitrary specification, and $S$ a set of answer clauses saturated up to redundancy from $\Lambda$ and abstracted.
We assume $\ansC{\square}{p}\notin S$ for any $p\in\mathcal{P}$.
Note that the set of all ground instances of $S$ might be unsatisfiable even if there is no computable solution for the specification, just an uncomputable one. Therefore, we cannot show refutational completeness of \syncal.
Instead, we show that 
from $S$, we can construct a counter-model for 
any computable program being a solution to $\Lambda$.
Formally, we show that given any ground computable program term $t[\bar{\alpha}]$ (possibly using $\iteF$), the formula $\cnf(\neg F[\bar{\alpha},t])$ is satisfiable. Towards this, we first eliminate $\iteF$ from answer clauses $\ansC{C_1}{t},\dots,\ansC{C_m}{t}$ coming from preprocessing, resulting in c-clauses.

\begin{definition}[$\iteF$ normal form]
\label{def:ite-normal-form}
Let $\AC=\ansC{C[y]}{y}$ be an answer clause and $t$ a program term. The \emph{$\iteF$ normal form} of $\AC$ w.r.t. $t$, denoted $\iteNF(\AC,t)$, is the set of c-clauses defined inductively as follows:
\begin{enumerate}
    \item If $C[y]$ contains $y$ and $t$ is of the form $\iteF(L,s_1,s_2)$, then it is
    $$\{\cC{D}{\lapp{\hat{L}}{\ELL}}\mid \cC{D}{\ELL}\in\iteNF(\AC,s_1)\}\cup\{\cC{D}{\lapp{L}{\ELL}}\mid \cC{D}{\ELL}\in\iteNF(\AC,s_2)\},$$
    \item otherwise, it is $\{\cCs{C[t]}\}$.\QED
\end{enumerate}
\end{definition}
Intuitively, we obtain a c-clause $\cC{C[s]}{\ELL}$ by elimination of $\iteF$ from $\ansC{C}{t}$ when $\ELL$ is the list of negations of $\iteF$-conditions needed to reach the branch term $s$ in $t$.

For any ground computable program term $t$, we construct a set of ground c-clauses $S'$ such that if $S'$ is satisfiable, then $\cnf(\neg F[\bar{\alpha},t])$ is also satisfiable. Then, we show that the total model $R_S'$ satisfies $S'$. The computable part of $S'$ is defined as follows.

\begin{definition}[Computable grounding]\label{def:comp-grounding}
We define the \emph{computable grounding} of $S$, denoted $\CGr(S)$, as the set of c-clauses that contains:
\begin{enumerate}
    \item all computable ground instances of c-clauses in $\iteNF(\AC,t')$ for all $\AC$ in the initial set of $\Lambda$ and program terms $t'$, and
    \item
    all computable ground c-clause instances of answer clauses in $S$.
    \QED
\end{enumerate}
\end{definition}
%
We prove that the model $R_{\CGr(S)}$ satisfies all computable clauses in $\CGr(S)$. The proof is very similar to standard completeness proofs, since the rules of \syncal for computable clauses correspond to the standard superposition rules: we do not apply any extra restrictions to inferences between computable clauses, and any inference between them results in a computable clause. By induction on $\succ$ over c-clauses, we obtain the following. 
\begin{restatable}{lemma}{computableSatisfied}\label{lemma:computable-clauses-completeness}
If $\ansC{\square}{p}\notin S$ for any computable $p\in\mathcal{P}$, then $R_{\CGr(S)}\models \CC$ for all $\CC\in \CGr(S)$.
\end{restatable}
%
\noindent
To obtain $S'$, we extend the set $\CGr(S)$ with uncomputable ground clauses. This set also depends on a computable program term $t$, and is called a grounding of $S$ w.r.t. $t$.
\begin{definition}[Grounding w.r.t. a program term]\label{def:uncomp-grounding}
Let $t\in\mathcal{P}$ be ground and $t'$ be the normal form of $t$ w.r.t. $R_{\CGr(S)}$. We define the \emph{grounding of $S$ w.r.t. program term $t$}, denoted $\Gr(S,t)$. In parallel, we define the \emph{derivation length} for each uncomputable $\CC$ in $\Gr(S,t)$, denoted $\DL(\CC)$.\\
We define $\Gr(S,t)$ as the minimal set of c-clauses containing $\CGr(S)$ such that:
\begin{enumerate}[itemsep=5pt]
    \item\label{def:uncomp-grounding:base} For any answer clause $\AC$ in the initial set of $\Lambda$, the set $\Gr(S,t)$ contains all ground instances of c-clauses in $\iteNF(\AC,t')$.\\ For any $\CC$ added to $\Gr(S,t)$ by this step, we define $\DL(\CC)= 0$.
    \item\label{def:uncomp-grounding:step} For any inference $\ansC{C_1}{p_1},\ldots,\ansC{C_n}{p_n}\vdash \ansC{C}{p}\sigma$ such that $\ansC{C_i}{p_i}\in S$ for all $1\le i\le n$, and $\ansC{C}{p}\in S$, if there is a substitution $\theta$ such that $\sigma\theta=\theta$, $\Gr(S,t)$ contains uncomputable c-clauses $\cC{C_1\theta}{\ELL_1}, \ldots, \cC{C_n\theta}{\ELL_n}$, and there is $j$, $1\leq j\leq n$, such that for all $1\le i\le n$, either $\ELL_i=\emptylist$, or $\ELL_i=\ELL_j$, then $\cC{C\theta}{\ELL_j}$ is in $\Gr(S,t)$.\\
    Take the inference and $\theta$ such that $\max(\DL(\cC{C_1\theta}{\ELL_1}), \ldots, \DL(\cC{C_n\theta}{\ELL_n}))$ is minimal. We define $\DL(\cC{C\theta}{\ELL_j})=\max(\DL(\cC{C_1\theta}{\ELL_1}), \ldots, \DL(\cC{C_n\theta}{\ELL_n}))+1$.\QED
\end{enumerate}
\end{definition}
In the remainder of this section, we assume that $t$ is an arbitrary but fixed ground computable program term. One key step in our main result is to show that inferences between uncomputable answer clauses with non-unifiable answers are not needed for a refutation that yields a computable answer. Towards this, we prove the following lemma about the grounding of $S$ w.r.t. $t$, using induction on the derivation length of c-clauses. 
\begin{restatable}{lemma}{uncomputableConditions}
\label{lemma:uncomputable-conditions} There is a set of lists of computable ground literals $\Delta$ that satisfies the following properties:
\begin{enumerate}
    \item\label{lemma:uncomputable-conditions:1} Any uncomputable $\CC\in\Gr(S,t)$ is of the form $\cC{C}{\ELL}$ where $\ELL\in\Delta\cup\{\emptylist\}$.
    \item\label{lemma:uncomputable-conditions:2} For any $\ELL\in\Delta\cup\{\emptylist\}$, there is a ground computable simple term $t_\ELL$ such that for any uncomputable $\cC{C'}{\ELL}\in\Gr(S,t)$, substitution $\theta$, and clause $\ansC{C}{p}\in S$ such that $C\theta=C'$, it holds that $p$ is a simple term and $p\theta=t_\ELL$ or $p$ is a variable not in $C$.
    \item\label{lemma:uncomputable-conditions:2'} If $\emptylist\notin\Delta$, then for any uncomputable $\cCs{C'}\in\Gr(S,t)$, substitution $\theta$, and clause $\ansC{C}{p}\in S$ such that $C\theta=C'$, it holds that $p$ is a variable not in $C$.
    \item\label{lemma:uncomputable-conditions:3} For any $\ELL,\KAY\in\Delta$ such that $\ELL\neq\KAY$, $\ELL$ and $\KAY$ contain a complementary literal.
\end{enumerate}
\end{restatable}
%
\noindent
Intuitively, any list of literals $\ELL$ in $\Delta$ corresponds to the negation of $\iteF$-conditions needed to reach some branch term, i.e. a simple term $s$ in $t$. The lemma states that:
\begin{enumerate}[label=(\roman*)]
    \item Any c-clause $\CC$ in $\Gr(S,t)$ represents an instance of an answer clause $\ansC{C}{p}$ in $S$ with either the empty condition, meaning that $p$ is any program term not depending on $C$, or with some conditions $\ELL$ in $\Delta$ such that $s$ is an instance of $p$.
    \item Any two lists of conditions in $\Delta$ contain a pair of complementary literals.
\end{enumerate}
This ensures that inferences between uncomputable c-clauses corresponding to different branches from $t$ are not needed, because if one of them is productive (\Cref{def:superposition-model}), the other one is necessarily satisfied by the model.
Now we can prove that $R_{\Gr(S,t)}$, obtained by extending $R_{\CGr(S)}$ with uncomputable clauses, is a model of $\Gr(S,t)$, and thus also of $\neg F[\bar{\alpha},t]$.
\begin{restatable}{lemma}{groundingSatisfied}\label{lemma:grounding-satisfied}
If $\ansC{\square}{p[\bar{\alpha}]}\notin S$ for any computable program term $p[\bar{\alpha}]$, then the formula $\neg F[\bar{\alpha},t]$ is satisfiable for any ground computable program term $t$.
\end{restatable}
%
\noindent
We conclude this section with the proof of our main result.
\begin{proof}[of~\Cref{theorem:completeness}]
Since $\Lambda$ is realizable, there is a program term $t[\bar{x}]$ with no variables except for $\bar{x}$, such that $\forall\bar{x}.F[\bar{x},t[\bar{x}]]$ is valid.
By contraposition of~\Cref{lemma:grounding-satisfied} it then follows that $\ansC{\square}{p[\bar{\alpha}]}\in S$ for an abstracted set $S$ saturated up to redundancy from $\Lambda$.
Since \syncal and $\Abs$ are sound (\Cref{th:soundness}), $p[\bar{x}]$ is a solution to $\Lambda$.
\QED
\end{proof}



\section{Related Work} \label{section:related_work}
\paragraph{Saturation-Based Synthesis.} A  saturation-based solution to the synthesis of recursion-free programs is introduced in~\cite{SynSat}. We show that~\cite{SynSat} is not complete. 
We introduce the \syncal calculus to simplify the reasoning of~\cite{SynSat}, as  (i) we assume all clauses to be answer clauses, and (ii) use most general unifiers instead of the arbitrary computable unifiers of ~\cite{SynSat} in inferences. Further, (iii) SUPRA implements a new abstraction rule $\Abs$, needed for completeness under realizability assumptions.

\paragraph{Deductive Approaches to Synthesis.}
Our work, as well as~\cite{SynSat,RecSynSat}, are based on the deductive synthesis approach of~\cite{MannaWaldinger1980},  adapted for resolution~\cite{LWC1974}. The mechanism in~\cite{10.1007/BFb0022265} restricts allowed programs, and subsumes our computability restrictions.

The work of~\cite{SynSat} extends~\cite{RecSynSat} with induction to support synthesis of recursive programs, with a proof of soundness and implementation, but without any completeness guarantees.
Our notion of realizability does not cover recursive programs, as their construction requires a mechanism for defining new functions not included in the computable symbols we work with.

A practical approach for synthesis in the presence of uncomputable symbols uses quantifier elimination and SMT solving~\cite{Z3Syn}, and requires $F[\bar{x}, y]$ to be quantifier-free in~\eqref{eq:spec} but may use linear arithmetic. The approach is complete only when $F[\bar{x}, y]$ uses theories that admit quantifier elimination, which excludes e.g. equalities or uninterpreted functions.

The SyGuS~\cite{sygus} format allows specifying programs in a fragment of first-order logic extended by theories and optionally a grammar for the sought programs. Some prominent SyGuS solvers are cvc5~\cite{SMTSyn} and DryadSynth~\cite{dryadsynth}, as evidenced by the SyGuS competition~\cite{sygus-comp}.
Component-based synthesis methods construct programs from logical specifications using a predefined set of functions.
The approach of~\cite{10.1007/3-540-58156-1_24} is deductive, while~\cite{GulwaniEtAl2011,TiwariEtAl2015} use SMT solving.
The GAPT framework~\cite{gapt} synthesizes programs from proofs, by computing witnesses of second-order formulas with quantifier elimination~\cite{SCAN}, or by extracting programs from natural deduction proofs in classical logic~\cite{friedman}.



\paragraph{Related Completeness Results.}
Our completeness proof is inspired by standard completeness proofs for superposition~\cite{10.1007/3-540-55253-7_22}. Yet, our model construction only works for groundings that substitute computable program terms into the specification due to realizability assumptions.
This relates our work to~\cite{10.1007/BFb0022265}. 
Finally, \cite{10.1007/BFb0023785} proved that a general synthesis system for recursive programs cannot exist.
This does not contradict our findings, as we claim completeness \emph{under the assumption that a program exists}.

\paragraph{Restricted Inferences.}
A rule similar to $\Abs$ is used in~\cite{hierachichtheories,10.1007/3-540-55253-7_22} to disallow superpositions into certain terms, while~\cite{DelayedUnification} introduces unification constraints to the clause level that are similar to abstracted terms. In~\cite{UWA-THI}, the unification algorithm is extended by abstraction to enhance theory reasoning within superposition, abstracting terms on demand to enable inferences while postponing expensive theory reasoning. 

\section{Summary and Outlook}
\label{section:summary}
We  
tackle program synthesis  and introduce the \syncal calculus for saturation-based proof search. Our calculus revises \cite{SynSat} 
by using an abstract unification rule, which unravels computable subterms from uncomputable literals such that uncomputable literals can be resolved without disobeying the computability constraint of the calculus. We also  use specific simplification orders and selection functions in \syncal, ensuring that uncomputable literals are always selected if there are any in a clause. 
\syncal changes the notion of clauses with answer literals to constrained clauses to make reasoning simpler. 
Not only do these adaptions retain soundness of \syncal, but they enable proving completeness of \syncal under the assumption that synthesis specifications have a solution. 
A natural direction for future work is investigating completeness under realizability assumption for synthesis of recursive functions using calculi that make use of induction.

\subsubsection{Acknowledgements.}
This research was funded in whole or in part by the ERC Consolidator Grant ARTIST 101002685, the Austrian Science Fund (FWF) 10.55776/DOC1345324, 
and  the  SBA Research  COMET Center SBA-K1 NGC managed  by the FFG.

\subsubsection{Disclosure of Interests.}
The authors have no competing interests to declare that are relevant to the content of this article.
 \bibliographystyle{splncs04}
 \bibliography{bibliography}

\ifbool{shortversion}{
 }{

\appendix

\newpage
\section*{Appendix}

\section{Proofs from \Cref{section:calculus}}\label{appendix:calculus}

\Soundness*
\begin{proof}\label{proof:soundness}
We prove soundness for the rules $\EqRes$ and $\Sup_C$.
For the other rules the proofs are analogous.

By $\cl(G)$ we denote the universal closure of the formula $G$.
When $v$ is a variable assignment, by $v(x)$ we denote the value it assigns to the value $x$. When $I$ is an interpretation, by $I(t)$ we denote the value by which it interprets the term $t$.

   The semantics of answer clauses were defined as follows: an interpretation $I$ satisfies an answer clause $\ansC{C}{p}$ w.r.t. a given specification $\Lambda$ if $I \models \cl(C\lor F[\bar{\alpha},p])$ for fresh constants $\bar{\alpha}$.
   For any rule of \syncal, 
\begin{prooftree}
\AxiomC{$\ansC{C_1}{p_1} \quad \dots \quad \ansC{C_n}{p_n}$} 
\RightLabel{,}
\UnaryInfC{$\ansC{C}{p}\sigma$}
\end{prooftree}
we have to show that if (i) $I \models \cl(C_1\lor F[\bar{\alpha},p_1]),\dots, I\models \cl(C_n\lor F[\bar{\alpha},p_n])$, then (ii) $I\models \cl(C\sigma\lor F[\bar{\alpha},p\sigma])$.
We prove this by contradiction: assume (i) holds but (ii) does not.
Then there exists a variable assignment $v$ such that $I$ extended by $v$, denoted $I_v$, does not satisfy $C\sigma\lor F[\bar{\alpha},p\sigma]$.
Therefore, $I_v\not\models C\sigma$ and $I_v\not\models F[\bar{\alpha},p\sigma]$.

Recall the rule $\EqRes$:
\begin{center}
\begin{tabular}{c c l}
\multirow{2}{*}{
\AxiomC{$\ansC{\sel{s\neqs t}\lor C}{p}$}
\LeftLabel{($\EqRes$)}
\UnaryInfC{$\ansC{C}{p}\sigma$}
\DisplayProof
}
& \multirow{2}{*}{where} & (1) $\sigma=\mgu(s,t)$\\
& & (2) $p\sigma$ is computable\\
\end{tabular}
\end{center}
For $\EqRes$, we assume (i) $I_u\models s\neqs t\lor C\lor F[\bar{\alpha},p]$ for any variable assignment $u$ and (ii) $I_v\not\models C\sigma$ and $I_v\not\models F[\bar{\alpha},p\sigma]$.
We construct a variable assignment $v'$ such that $v'(x)=I_v(x\sigma)$ for each variable $x$ occurring in $\ansC{s\neqs t\lor C}{p}$, and $v'(z)=v(z)$ for all other variables $z$.
Since $\sigma$ unifies $s$ and $t$, we get $I_{v'}\not\models s\neqs t$.
Further, from $I_v\not\models C\sigma$ we get $I_{v'}\not\models C$.
Hence, from (i) follows that $I_{v'}\models F[\bar{\alpha},p]$.
Now, in~\Cref{def:semantics} we assumed that $F[\bar{\alpha}, y]$ and $\ansC{C}{p}$ do not have any common variables. Therefore $I_{v'}(p) = I_v(p\sigma)$, and thus from (ii) and the definition of $v'$ we get $I_{v'}\not\models F[\bar{\alpha}, p]$.
We derived a contradiction, $\EqRes$ is therefore sound.

Recall now the rule $\Sup_C$:
\begin{center}
\begin{tabular}{c c l}
\multirow{4}{*}{
\AxiomC{$\sel{l\eqs r}\lor \ansC{C}{p}$}
\AxiomC{$\ansC{\sel{s[l']\deqs t}\lor D}{q}$}
\LeftLabel{($\Sup_C$)}
\BinaryInfC{$\ansC{s[r]\deqs t\lor C \lor D}{\iteF(l\eqs r,q,p)}\sigma$}
\DisplayProof
}
& \multirow{4}{*}{where} & (1) $\sigma=\mgu(l,l')$\\
& & (2) $l'$ is not a variable\\
& & (3) $r\sigma \nsucceq l\sigma$, $t\sigma  \nsucceq s[l']\sigma$ \\
& & (4) $\mathsf{ite}(l\eqs r,q,p)\sigma$ is computable \\
\end{tabular}
\end{center}
For $\Sup_C$, we assume (i) $I_u\models l\eqs r \lor C\lor F[\bar{\alpha},p]$, $I_u\models s[l']\deqs t \lor D\lor F[\bar{\alpha},q]$ for any variable assignment $u$ and (ii) $I_v\not\models (s[r]\deqs t\lor C \lor D)\sigma$ and $I_v\not\models F[\bar{\alpha},\iteF(l\eqs r,q,p)\sigma]$ for some variable assignment $v$.
We construct a variable assignment $v'$ such that $v'(x)=I_v(x\sigma)$ for each variable $x$ occurring in 
$l\eqs r\lor \ansC{C}{p}$ and $\ansC{s[l']\deqs t\lor D}{q}$, and $v'(z)=v(z)$ for all other variables $z$.
Then from $I_v\not\models (s[r]\deqs t\lor C \lor D)\sigma$ we get $I_{v'}\not\models s[r]\deqs t\lor C \lor D$. 
We have two cases, either $I_{v'}\models l \eqs r$ or $I_{v'}\not\models l \eqs r$. Assume the first, then from $I_{v'}\not\models s[r]\deqs t$, also $I_{v'}\not \models s[l']\deqs t$ (the evaluation of $l$ and $l'$ under $v'$ is the same since $\sigma$ is a unifier of $l$ and $l'$), hence (i) implies $I_{v'}\models F[\bar{\alpha},q]$. We evaluate $I_{v'}(\iteF(l\eqs r,q,p)=I_{v'}(q) = I_v(q\sigma)$, and note that we assumed that answer clauses $\ansC{C}{p}$ do not have common variables with $F[\bar{\alpha},y]$, thus from (ii) we get $I_{v'}\not\models F[\bar{\alpha}, q]$, contradiction.

Similarly, the second case, $I_{v'}\not\models l \eqs r$, implies $I_{v'}\models F[\bar{\alpha},p]$. From $I_{v'}(\iteF(l\eqs r,q,p)=I_{v'}(p) = I_v(p\sigma)$ and (ii) we derive the contradiction $I_{v'}\not\models F[\bar{\alpha}, p]$.

Hence, $\Sup_C$ is sound.
\CommentedOut{
\PH{I think the following does not work, see my todo-notes.}

   The semantics of answer clauses were defined as follows: an interpretation $I$ and variable assignment $v$ (denoted by $I_v$) satisfy an answer clause $\ansC{C}{p}$ w.r.t. a given specification $\Lambda$ if either (i) $I_v \models C$ or (ii) $I_v \models F[\bar{\alpha},p]$. 
   For any rule of \syncal, 
\begin{prooftree}
\AxiomC{$\ansC{C_1}{p_1} \quad \dots \quad \ansC{C_n}{p_n}$} 
\RightLabel{,}
\UnaryInfC{$\ansC{C}{p}\sigma$}
\end{prooftree}
if for some interpretation $I$ and variable assignment $v$ all premises fulfill (i) then also $I_v \models C\sigma$ by the soundness of the Superposition Calculus.\todo{PH: this does not hold. The soundness of the sup. calculus says that if for all $v$ we have $I_v$ satisfies the premises, then for all $v'$ we have $I_{v'}\models C\sigma$. It doesn't work for just one $v$.} In particular, the rules $\EqFac$, $\EqRes$, and $\Sup_C$ have additional conditions about the computability of answers, and $\Sup_U$ has a more restrictive unifier, which in both cases do not affect soundness. 

What remains to show is that if one of the answer clauses $\ansC{C_i}{p_i}$ satisfies (ii) for an interpretation $I$ and variable assignment $v$ then it follows that the derived answer clause $\ansC{C}{p}$ is satisfied by $I_v$ w.r.t. $\Lambda$. 

For $\EqRes$ and $\EqFac$ this immediately follows, $p\sigma$ is an instance of $p$.\todo{PH: I think this isn't true because the variables were assigned some values by $v$, so from $I_v\models p$ does not follow $I_v\models p\sigma$...} Similar arguing works for $\Sup_U$ since $\sigma$ unifies both answers $p$ and $q$.

For the inference rule $\Sup_C$ we have to consider three cases.

Assume first that both the left clause, $l\eqs r\lor \ansC{C}{p}$, and the right clause $\ansC{s[l']\deqs t\lor D}{q}$ do not satisfy (i) but (ii). Then $I_v\models F[\bar{\alpha},\iteF(l\eqs r,q,p)\sigma]$.

Consider further that the left clause, $l\eqs r\lor \ansC{C}{p}$ does not satisfy (i) but (ii) and the right clause $\ansC{s[l']\deqs t\lor D}{q}$ satisfies (i). If $I_v \models D$ we are done, so assume $I_v \models s[l']\deqs t$. Then either $I_v \nmodels (l\eqs r)\sigma$ in which case $I_v\models F[\bar{\alpha},\iteF(l\eqs r,q,p)\sigma]$ or $I_v\models (l\eqs r)\sigma$ then also $I_v\models (s[r]\deqs t)\sigma$.

Lastly we consider the case where the left clause, $l\eqs r\lor \ansC{C}{p}$ satisfies (i) and the right clause $\ansC{s[l']\deqs t\lor D}{q}$ does not satisfy (i) but (ii). If $I_v \models C$ also $I_v \models C\sigma$, so we can assume that $I_v \models l\eqs r$. Then we know $I_v\models F[\bar{\alpha},\iteF(l\eqs r,q,p)\sigma]$.
}
\QED
\end{proof}

\Cor*
\begin{proof}\label{proof:corollary}
We first show that all answer clauses in the initial set $S_0$ from $\Lambda$ hold (\Cref{def:semantics}).
The initial set $S_0$ is $\ansC{\cnf(\neg F[\bar{\alpha}, y])}{y}$.
Each answer clause in $S_0$ is thus of the form $\ansC{C}{y}$, where $C\in\cnf(\neg F[\bar{\alpha}, y])$.

We will show that $C\lor F[\bar{\alpha}, y]$ is true in any interpretation $I$ under any variable assignment $v$.
The formula $\neg F[\bar{\alpha}, y]$ is either true or false in $I$ under $v$.
If $\neg F[\bar{\alpha}, y]$ is true in $I$ under $v$, and $C\in\cnf(\neg F[\bar{\alpha}, y])$, also $C$ is true in $I$ under $v$. Thus, $C\lor F[\bar{\alpha},y]$ is true too.
On the other hand, if $\neg F[\bar{\alpha}, y]$ is false in $I$ under $v$, it means that its negation $F[\bar{\alpha}, y]$ is true in $I$ under $v$.
Then $C\lor F[\bar{\alpha},y]$ is true too.
It follows that $C\lor F[\bar{\alpha},y]$ is true in any interpretation and any variable assignment, therefore it holds, and thus also $\ansC{C}{y}$ holds.

We have shown that we start saturation with a set of valid answer clauses.
Since each inference by \syncal is sound by~\Cref{th:soundness}, all answer clauses derived in saturation are also valid.
Then, if $\ansC{\square}{p[\bar{\alpha}]}$ is derived, it is valid, meaning that $\square\lor F[\bar{\alpha},p[\bar{\alpha}]]$ holds.
Further, $\square\lor F[\bar{\alpha},p[\bar{\alpha}]]$ is equivalent to $F[\bar{\alpha},p[\bar{\alpha}]]$, which thus also holds.
Finally, since $\bar{\alpha}$ were fresh Skolem constants, it follows that $\forall\bar{x}.F[\bar{x},p[\bar{x}]]$ is valid too, meaning that $p[\bar{x}]$ is a solution for $\Lambda$.\QED
\end{proof}
   
\absSound*
\begin{proof}\label{proof:absSound}
   The soundness of $\Abs$ is given as follows. Let $I$ be an interpretation that satisfies $\cl(\ansC{s[k]\deqs t\lor C}{p})$. If $I$ satisfies $\cl(s[k]\deqs t)$ but not $\cl(x\neqs k)$ it satisfies $\cl(s[x]\deqs t)$. Conversely, if it does not satisfy $\cl(s[x]\deqs t)$ it satisfies $\cl(x\neqs k)$. If $I$ satisfies $\cl(C)$ or $\cl(F[\bar{\alpha},p])$ the conclusion trivially holds.  \QED
\end{proof}

\absConditions*
\begin{proof}
Suppose condition (2) holds.
From $k\theta$ being computable follows that $k$ is computable too and thus (2a) holds.
Further, $s[k]\theta$ is uncomputable.
Either $s[k]$ is uncomputable too, or the uncomputable symbol was introduced as a part of some term $u$ by $\{y\mapsto u\}\subseteq \theta$ for some $y$.
In the latter case, since $k\theta$ is computable, we know that $y$ must occur in $s[k]$ but not in $k$. Thus (2b) holds too. 

Suppose now conditions (2a) and (2b) hold.
If $s[k]$ is uncomputable, the empty substitution satisfies condition (2).
Otherwise $s[k]$ contains some variable $y$ not occurring in $k$. Then $\theta = \{y\mapsto u\}$ for an arbitrary uncomputable term $u$ satisfies condition (2).\QED
\end{proof}

\section{Orderings}
\label{appendix:ordering}

\subsection{Orderings and selection} \label{subsection:orderingselection}
The following lemmas state that for these types of LPOs any uncomputable term, literal or clause cannot be smaller than or equal to any computable term, literal or clause respectively. 


\begin{restatable}{lemma}{lpoLemma}\label{lem:lpo}
For a signature $\Sigma=\Sigma_c \uplus \Sigma_u$, an LPO with a precedence relation $\gg$, such that every $k \in \Sigma_u$ has higher precedence than any $k' \in \Sigma_c$, and terms $s \in \mathcal{T}_c,\;\; t \in \mathcal{T}_u$ it holds that $s\nsucceq_{lpo}t$.
\end{restatable}
\begin{proof}    
     Note that $s \neq t$, because $t$ contains an uncomputable symbol while $s$ does not. We still need to prove that $s \nsucc_{\lpo} t$. We prove this by induction over the term $t$.  Term $t$ cannot be a variable, so case \ref{prop:lpo1} in the definition of LPO is trivially never fulfilled. 
     \paragraph{Base case.} The topmost symbol of term $t$ is uncomputable; $\topsym(t) \in \Sigma_u$\footnote{By $\topsym(t)$ we denote the topmost symbol of term $t$.}. Assume by contradiction that $s \succ_{\lpo} t$. Let $s'$ be smallest subterm of $s$ such that $s' \succ_{\lpo} t$. By definition $\topsym(t) \gg \topsym(s')$, neither case \ref{prop:lpo2} nor \ref{prop:lpo3} can hold, contradiction.
\paragraph{Induction step.}
     Let $t=f(t_1,\dots t_{j-1},t_j,t_{j+1}\dots,t_n)$ where $\topsym(t_j) \in \Sigma_u$. Assume by contradiction that $s \succ_{\lpo} t$. Let $s'$ be smallest subterm of $s$ such that $s' \succ_{\lpo} t$. Because $s'\nsucc_{\lpo} t_j$ neither case \ref{prop:lpo3} nor \ref{prop:lpo4} can hold, contradiction.\QED
\end{proof}

The next Lemma immediately highlights that we don't lose the property from Lemma~\ref{lem:lpo} using bag extensions.
\begin{restatable}{lemma}{bagorder}\label{lem:bagorder}
Assume that for two bags $\{x_1,\dots,x_n\}$ and $\{y_1,\dots,y_m\}$ there is $j \in \{1,\dots,m\}$ such that for any $i \in \{1,\dots,n\}$ it holds that $x_i \nsucceq y_j$. Then $\{x_1,\dots,x_n\}\nsucceq^{bag}\{y_1,\dots,y_m\}$.    
\end{restatable}
\begin{proof}
    Assume by contradiction that $\{x_1,\dots,x_n\}\succeq^{bag} \{y_1,\dots,y_m\}$. By definition of the bag extension $\{y_1,\dots,y_m\}\succeq^{bag}\{y_j\}$. By transitivity this implies $\{x_1,\dots, x_n\}\succeq^{bag} \{y_j\}$. But this can only be if there is $i \in \{1,\dots,n\}$ such that $x_i \succeq y_j$. By assumption this cannot be the case. \QED  
\end{proof}

The next two Lemmas follow very easily from Lemma~\ref{lem:bagorder}.
\begin{restatable}{lemma}{bagordertwo}\label{lem:bagorder2}
  Let $L_1,L_2$ be two literals, where $L_1$ is uncomputable and $L_2$ is computable. By $\succ$ we denote the bag extension of an LPO with a precedence relation $\gg$ such that every $k \in \Sigma_u$ has higher precedence than any $k' \in \Sigma_c$. Then $L_2 \nsucceq L_1$, i.e. for an uncomputable clause $C$ at least one uncomputable literal in $C$ is maximal.  
\end{restatable}
\begin{proof}
    We have already shown that uncomputable terms cannot be smaller  than computable terms. Therefore by Lemma~\ref{lem:bagorder} $L_2 \nsucceq L_1$. \QED
\end{proof}

\begin{restatable}{lemma}{bagorderthree}\label{lem:bagorder3}
  Let $C_1,C_2$ be two clauses, where $C_1$ is uncomputable and $C_2$ is computable. By $\succ$ we denote the bag extension of literals of the bag extension of an LPO with a precedence relation $\gg$ such that every $k \in \Sigma_u$ has higher precedence than any $k' \in \Sigma_c$. Then $C_2 \nsucceq C_1$.  
\end{restatable}
\begin{proof}
    We have already shown that uncomputable literals cannot be smaller  than computable literals. Therefore by Lemma~\ref{lem:bagorder2} $C_2 \nsucceq C_1$. \QED
\end{proof}

\subsection{Knuth-Bendix orders}

In the following we give an extended definition for \emph{Knuth-Bendix orders} that can also be used with \syncal. A \emph{transfinite weight function} is a function $w$ from $\Sigma = \Sigma_c \uplus \Sigma_u$ to the union of the ordinals $\omega \cup \omega_1$, where $\omega:=\{0,1,2,\dots\}$ and $\omega_1=\{\omega + 1, \omega +2, \dots\}$ such that  $w(\Sigma_c) \subseteq \omega$ and $w(\Sigma_u) \subseteq \omega_1$\footnote{we denote with $f(A)$ the image of a function, where $f:X \to Y$ and $A \subseteq X$.}. We will refer to $w(f)$ as the \emph{weight} of $f$. We denote by $w_0$ the smallest weight of constants.
%
%
For $p\in\Sigma\cup\mathcal{V}$, we write $|t|_p$ to denote  the \emph{number of occurrences} of $p$ in 
term $t$. For example, $|f(x,x)|_f=1$, $|f(x,x)|_x=2$ and $|f(x,x)|_y=0$. Let $\mathcal{P}(\mathcal{V})$ be the set of linear expressions over $\mathcal{V}$ with integer coefficients. The \emph{weight of a term} $t$, denoted by $|t|$, is a linear expression in $\mathcal{P}(\mathcal{V})$ defined as:
$$|t|:=\sum_{f\in\Sigma}|t|_f\cdot w(f) + \sum_{x\in\mathcal{V}}|t|_x\cdot x,$$
where arithmetic (i.e. $+,\cdot, <$) on ordinal numbers is defined in the standard way.
A substitution $\sigma$ can also be considered as mapping from linear expressions to linear expressions, as follows: 
\[
\sigma(\alpha_0 + \alpha_1\cdot x_1 + \ldots + \alpha_n\cdot x_n):=\alpha_0 + \alpha_1\cdot |x_1\sigma| + \ldots + \alpha_n\cdot |x_n\sigma|.
\]
For example, if $w(f)=2$, $w(a)=1$ and $\sigma=\{x\mapsto a\}$, then $|f(x,x)|=2\cdot x+2$ and $\sigma(|f(x,x)|)=\sigma(2\cdot x+2)=2\cdot|x\sigma|+2=4$.
It is not hard to argue that $|t\sigma|$ = $\sigma(|t|)$. 
Let $e\in\mathcal{P}(\mathcal{V})$ be a linear expression. We call a substitution $\sigma$ \emph{grounding} for $e$ if $\sigma(e)$ does not contain variables. 
We write $e_1 > e_2$ if $\sigma(e_1) > \sigma(e_2)$  for all grounding substitutions $\sigma$ for $e_1$ and $e_2$. We write $e_1 \gtrsim e_2$ if $\sigma(e_1) \geq \sigma(e_2)$ for all grounding substitutions $\sigma$ for $e_1$ and $e_2$.
The \emph{transfinite Knuth-Bendix order (tKBO)}\cite{tfkbo}\footnote{Note that this definition is only a specific instance of the family of transfinite KBOs given in \cite{tfkbo} but similar properties hold.}, denoted by $\succ_\tkbo$, is parameterized by a precedence relation $\gg$ and a transfinite weight function $w$. For terms $s,t$, we have  $s \succ_\tkbo t$ if:
\begin{enumerate}
\item\label{prop:kbo1} $|s|>|t|$, or
\item\label{prop:kbo2} $|s|\gtrsim|t|$, $s = f(s_1,...,s_n)$, $t = g(t_1,...,t_m)$ and $f\gg g$, or
\item\label{prop:kbo3} $|s|\gtrsim|t|$, $s = f(s_1,...,s_n)$, $t = f(t_1,...,t_n)$ and there exists $1\le i\le n$ such that $s_i\succ_\tkbo t_i$ and
$s_j=t_j$ for all $1\le j<i$.
\end{enumerate}

The transfinite KBO is a simplification order for any precedence relation $\gg$ and transfinite weight function $w$, if $w_0>0$ and $f\gg g$ for all $g\in\Sigma$ different from $f$, for all unary $f\in\Sigma$ with $w(f)=0$ \cite{tfkbo}. \\
\begin{restatable}{lemma}{kboLemma}\label{lem:kbo}
For a signature $\Sigma=\Sigma_c \uplus \Sigma_u$, a transfinite KBO on $\Sigma$ and terms $s \in \mathcal{T}_c,\;\; t \in \mathcal{T}_u$ it holds that $s\nsucceq_{tkbo}t$.
\end{restatable}
\begin{proof}
    The term $t$ contains an uncomputable symbol, while $s$ does not, therefore $s \neq t$. Assume further that $s \succ_{\tkbo} t$. This implies that $|s|\gtrsim|t|$. Therefore for any grounding substitution $\sigma$ it holds that $\sigma(|s|)\geq \sigma(|t|)$. In particular $\sigma(|s|)\geq \sigma(|t|)$ for any substitution $\sigma$ that maps variables to only computable terms. This cannot be the case because by assumption $\sigma(|s|)\in \omega_0$ and $\sigma(|t|) \in \omega_1$. \QED
\end{proof}
\begin{remark}
    In particular computable expressions (literals or clauses) can never be bigger than uncomputable expressions using transfinite KBOs, see Lemma~\ref{lem:bagorder},\ref{lem:bagorder2},\ref{lem:bagorder3}.
\end{remark}

\section{Proofs from \Cref{section:completeness}}\label{appendix:completeness}
\Rproperties*
\begin{proof}

Towards the first property, first, we prove that the rewrite relation of $R_S$ is terminating. $R_S$ terminates since $l \succ r$ for all its rules $l \to r$ by condition \ref{productive:ordered} of Definition~\ref{def:superposition-model}.

Second, we prove that the rewrite relation of $R_S$ is confluent. Newman's lemma states that every terminating and locally confluent term rewriting system is confluent~\cite{NewmansLemma}. Using Newman's lemma, it suffices to show local confluence. Since all rules in $R_S$ are ground, it is enough to show that the rules are non-overlapping, i.e. that there are no two different rules $l \to r$ and $l' \to r'$ in $R_S$ where $l'$ is a subterm of $l$. Let $\CC$ be a c-clause that produces a rule $l \to r$. By condition \ref{productive:irreducible} of Definition~\ref{def:superposition-model}, $R_S^\CC$ cannot contain any rule $l' \to r'$ s.t. $l'$ is a subterm of $l$. Suppose that some $l'\to r'$ is produced by a c-clause $\CD\succ\CC$ s.t. $l'$ is a subterm of $l$. Since $l'$ is maximal in $\CD$ (condition \ref{productive:maximal} and \ref{productive:ordered} of Definition~\ref{def:superposition-model}), we have $l'\succ l$, hence $l'$ cannot be a subterm of $l$ (using the subterm property of $\succ$).

Towards the second property, suppose that $R_S^\CC\models \CC$ due to a positive literal $s\eqs t$ in $\CC$ with $R_S^\CC\models s\eqs t$. For any clause $\CD\succ \CC$, we have $R_S\supseteq R_S^\CD\supseteq R_S^\CC$. This implies that $R_S^\CD\models s\eqs t$ and $R_S\models s\eqs t$, hence $R_S^\CD\models \CC$ and $R_S\models \CC$. Otherwise, there is a negative literal $s\neqs t$ in $\CC$ and $R_S^\CC\nmodels s\eqs t$. W.l.o.g. assume $s\succ t$. For any clause $\CD\succ \CC$ that produces a rewrite rule $l\to r$, $l$ cannot be a subterm of $s$, hence no rule in $R_S^\CD$ (or $R_S$) can rewrite $s$. Since $R_S$ is convergent (by property 1), this means that $R_S\nmodels s\eqs t$.\QED
\end{proof}

\computableSatisfied*

\begin{proof}
The proof is by contradiction and by induction on $\succ$ over c-clauses. Suppose there is a c-clause $\CC\in \CGr(S)$ such that (i) $R_{\CGr(S)}\nmodels \CC$ and (ii) $\CC$ is the minimal such c-clause in $\CGr(S)$ w.r.t. $\succ$. We consider the following cases. \footnote{Note that in all cases we assume that all preceding cases do not hold, to make our case distinction simpler. Similarly in~\Cref{lemma:grounding-satisfied}.}

\begin{proofcase}{1}{$\CC$ is of the form $\cC{C'}{\ELL}$ where $\ELL$ is not empty.}\end{proofcase}

\noindent
Let $\CD$ be the c-clause $\cCs{C'}$, which is also in $\CGr(S)$. Moreover, $\CC\succ \CD$ and since $C'$ is false, $R_{\CGr(S)}\nmodels\CD$. Hence we have found a smaller false c-clause in $\CGr(S)$, contradiction.

\begin{proofcase}{2}{$\CC$ is of the form $\cCs{C'}$ and there is an answer clause $\ansC{C}{p}$ in $S$ and a substitution $\theta$ s.t. $C\theta=C'$, and there is a variable $x$ in $C$ such that $x\theta$ is reducible in $R_{\CGr(S)}$.}\end{proofcase}

\noindent
Let $\theta'$ be the substitution $\theta$, except that it maps $x$ to the normal form of $x\theta$ w.r.t. $R_{\CGr(S)}$, and let $\CD$ be the c-clause $\cCs{C\theta'}$. We have $\CD\in \CGr(S)$ and since $x\theta\succ x\theta'$ and $C$ contains $x$, we have $\CC\succ \CD$. Also, $x\theta'$ is the normal form of $x\theta$, so $R_{\CGr(S)}\nmodels \CD$ either. Thus, we have found a false clause smaller than $\CC$ in $\CGr(S)$, contradiction.

\begin{proofcase}{3}{$\CC$ is of the form $\cCs{\underline{s'\neqs t'}\lor C'}$ where $s'\neqs t'$ is maximal in $s'\neqs t'\lor C'$.}\end{proofcase}

\begin{proofsubcase}{3.1}{$s'=t'$.}\end{proofsubcase}

\noindent
Then there is an answer clause $\ansC{\underline{s\neqs t}\lor C}{p}$ in $S$, a ground substitution $\theta$ such that $(\underline{s\neqs t}\lor C)\theta=\underline{s'\neqs t'}\lor C'$. Consider the following inference:
\begin{prooftree}
    \AxiomC{$\ansC{\underline{s\neqs t}\lor C}{p}$}
    \LeftLabel{($\EqRes$)}
    \UnaryInfC{$\ansC{C}{p}\sigma$}
\end{prooftree}
where $\sigma=\mgu(s,t)$. Let $\CD$ be the c-clause $\cCs{C'}$. By $s\theta=t\theta=s'$, we have that $\theta$ is a unifier of $s$ and $t$. Since $p$ and $\CC$ are computable, the unifier $\sigma$ only contains computable symbols in its range, hence $p\sigma$ is also computable, satisfying condition 2 of $\EqRes$. Therefore, $\ansC{C}{p}\sigma\in S$ and thus $\CD\in\CGr(S)$. Due to $\{\{s',s',s',s'\}\}\uplus C'\succbag C'$, we also have $\CC \succ \CD$. Also, $R_{\CGr(S)}\nmodels C'$, so $R_{\CGr(S)}\nmodels \CD$. Thus, we have found a smaller false c-clause in $\CGr(S)$, contradiction.

\begin{proofsubcase}{3.2}{W.l.o.g. $s'\succ t'$\footnote{Note that $\succ$ is total on ground terms.} and $s'$ is reducible by some $l'\to r'\in R_{\CGr(S)}$.}\end{proofsubcase}

\noindent
Then $s'$ is of the form $s'[l']$ and there is a rule $l'\to r'$ in $R_{\CGr(S)}$ produced by some computable clause $\cC{l'\eqs r'\lor D'}{\KAY}$. Note that we have $\KAY=\emptylist$, as $\cC{l'\eqs r'\lor D'}{\KAY}$ is smaller than all clauses $\cC{l'\eqs r'\lor D'}{\KAY'}$ where $\KAY'$ is non-empty, therefore it must produce the rule $l'\to r'$. Then there are clauses $\ansC{\underline{s[k]\neqs t}\lor C}{p}$ and $\ansC{\underline{l\eqs r}\lor D}{q}$ in $S$ such that $k$ is not a variable (see Case 2), and a ground substitution $\theta$ such that $l\theta=k\theta=l'$, $(s[k]\neqs t\lor C)\theta=s'\neqs t'\lor C'$ and $(l\eqs r\lor D)\theta=l'\eqs r'\lor D'$. Consider the following inference:
\begin{prooftree}
    \AxiomC{$\ansC{\underline{l\eqs r}\lor D}{q}$}
    \AxiomC{$\ansC{\underline{s[k]\neqs t}\lor C}{p}$}
    \LeftLabel{($\Sup_C$)}
    \BinaryInfC{$\ansC{s[r]\neqs t\lor C\lor D}{\iteF(l\eqs r,p,q)}\sigma$}
\end{prooftree}
where $\sigma=\mgu(l,k)$. Since $l'\eqs r'\lor D'$ is productive, \Cref{productive:ordered} from \Cref{def:superposition-model} and the assumption $s'\succ t'$ imply that $r\sigma\nsucceq l\sigma$ and $t\sigma\nsucceq s[k]\sigma$ (condition 3 of the $\Sup_C$ rule). Also, since both c-clauses (and therefore both answer clauses) are computable, $\iteF(l\eqs r,p,q)\sigma$ must be also computable, satisfying condition 4 of the $\Sup_C$ rule. Then, $(s[r]\neqs t\lor C\lor D)\sigma\in S$. Let $\CD$ be the c-clause $\cCs{s'[r']\neqs t'\lor C'\lor D'}$. Note that $\CD\in\CGr(S)$.
Further, we have
$$\{\{s'[l'],s'[l'],t',t'\}\}\uplus C'\succbag \{\{s'[r'],s'[r'],t',t'\}\}\uplus C'\uplus D'$$
and hence $\CC \succ \CD$. Also, $R_{\CGr(S)}\nmodels \CD$: (1) by assumption, $C'$ and $D'$ are false in $R_{\CGr(S)}$, and (2) $s'[l']$ has the same normal form as $t'$ and $s'[r']$, therefore $t'$ and $s'[r']$ have the same normal form too by confluence. Thus, we have found a false c-clause smaller than $\CC$ in $\CGr(S)$, contradiction.

\begin{proofsubcase}{3.3}{W.l.o.g. $s'\succ t'$ and $s'$ is irreducible.}\end{proofsubcase}

\noindent
Then, since $s'\succ t'$ and $s'$ is in normal form, $R_{\CGr(S)}\models s'\neqs t'$ and $\CC$ is satisfied, contradiction.

\begin{proofcase}{4}{$\CC$ is of the form $\cCs{\underline{s'\eqs t'}\lor C'}$ where $s'\eqs t'$ is maximal in $s'\eqs t'\lor C'$.}\end{proofcase}

\noindent
W.l.o.g. $s'\succ t'$.

\begin{proofsubcase}{4.1}{$s'$ is reducible by $l'\to r'\in R_{\CGr(S)}$.}\end{proofsubcase}

\noindent
This is similar to Case 3.2. Then $s'$ is of the form $s'[l']$ and $l'\to r'$ is produced by some c-clause $\cCs{l'\eqs r'\lor D'}$. There are clauses $\ansC{\underline{s[k]\eqs t}\lor C}{p}$ and $\ansC{\underline{l\eqs r}\lor D}{q}$ in $S$, such that $k$ is not a variable (see Case 2), a ground substitution $\theta$ such that $l\theta=k\theta=l'$, $(s[k]\eqs t\lor C)\theta=s'[l']\eqs t'\lor C'$ and $(l\eqs r\lor D)\theta=l'\eqs r'\lor D$. Consider the following inference:
\begin{prooftree}
    \AxiomC{$\ansC{\underline{l\eqs r}\lor D}{q}$}
    \AxiomC{$\ansC{\underline{s[k]\eqs t}\lor C}{p}$}
    \LeftLabel{($\Sup_C$)}
    \BinaryInfC{$\ansC{s[r]\eqs t\lor C\lor D}{\iteF(l\eqs r,p,q)}\sigma$}
\end{prooftree}
where $\sigma=\mgu(l,k)$. Let $\CD$ be the c-clause $\cCs{s'[r']\eqs t'\lor C'\lor D'}$. We check the conditions the $\Sup_C$ rule similarly as in Subcase 3.2. Therefore, $\ansC{s[r]\lor C\lor D}{\iteF(l\eqs r,p,q)}\sigma\in S$ and $\CD\in \CGr(S)$. We also have the following
$$\{\{s'[l'],t'\}\}\uplus C'\succbag \{\{s'[r'],t'\}\}\uplus C'\uplus D'$$
and hence $\CC \succ \CD$. Also, $R_{\CGr(S)}\nmodels \CD$: (1) by assumption, $C'$ and $D'$ are false in $R_{\CGr(S)}$, and (2) $s'[l']$ have distinct normal forms with $t'$, so $s'[r']$ and $t'$ must have distinct normal forms too (otherwise $\CC$ would be true in $R_{\CGr(S)}$ by confluence). Thus, we have found a false clause in $S$ smaller than $\CC$.\medskip

\begin{proofsubcase}{4.2}{$s'$ is irreducible in $R_{\CGr(S)}$.}\end{proofsubcase}

\noindent
Then, $C'$ is of the form $s'\eqs r'\lor D'$ such that $t'\succeq r'$ and $s'\eqs r'$ is true in $R_{\CGr(S)}^{\prec \CC}\cup\{s'\to t'\}$ (i.e. \Cref{productive:false3} of \Cref{def:superposition-model} is violated). Note that this also covers the case when \Cref{productive:maximal} of \Cref{def:superposition-model} is violated, as then $t'=r'$ and $\CC$ is of the form $\cCs{\underline{s'\eqs t'}\lor \underline{s'\eqs t'}\lor D'}$.

Then there is an answer clause $\ansC{\underline{s\eqs t}\lor l\eqs r\lor D}{p}$ in $S$, a ground substitution $\theta$ such that $(s\eqs t\lor l\eqs r\lor D)\theta=s'\eqs t'\lor s'\eqs r'\lor D'$. Consider the following inference:
\begin{prooftree}
    \AxiomC{$\ansC{\underline{s\eqs t}\lor l\eqs r\lor D}{p}$}
    \LeftLabel{($\EqFac$)}
    \UnaryInfC{$\ansC{s\eqs t\lor t\neqs r\lor D}{p}\sigma$}
\end{prooftree}
where $\sigma=\mgu(s,l)$. Let $\CD$ be the c-clause $s'\eqs t'\lor t'\neqs r'\lor D'$. Since $s'\succ t'$ and $t'\succeq r'$ by assumption, condition 2 of the $\EqFac$ rule holds. Also, since $p$ and $\CC$ are computable, the unifier $\sigma$ only contains computable symbols in its range, hence $p\sigma$ is also computable, satisfying condition 3 of $\EqFac$. Then, $\ansC{s\eqs t\lor t\neqs r\lor D}{p}\sigma\in S$ and $\CD\in \CGr(S)$. By $s'\succ t'\succeq r'$, we have the following
$$\{\{s',t'\},\{s',r'\}\}\uplus D'\succbag \{\{s',t'\},\{t',t',r',r'\}\}\uplus D'$$
and hence $\CC \succ \CD$. Also, $R_{\CGr(S)}\nmodels (s'\eqs t'\lor t'\neqs r'\lor D')$: (1) by assumption, $s'\eqs t'$ and $D'$ are false in $R_{\CGr(S)}$, and (2) since $s'\eqs r'$ is true in $R_{\CGr(S)}^{\prec\CC}\cup\{s'\to t'\}$, we have that $R_{\CGr(S)}\models t'\eqs r'$ and therefore $t'\neqs r'$ is false in $R_{\CGr(S)}$. Thus, we have found a false clause in $\CGr(S)$ smaller than $\CC$.\QED

\end{proof}

\begin{lemma}\label{lem:complementaryConditions}
Given a computable ground program term $t$, there is a set of lists of computable literals $\Delta$ that satisfies the following properties:
\begin{enumerate}
    \item\label{lem:complementaryConditions:1} For any answer clause $\AC$ of the form $\ansC{C[x]}{x}$, any uncomputable c-clause in $\iteNF(\AC,t)$ is of the form $\cC{C}{\ELL}$, where $\ELL\in\Delta\cup\{\emptylist\}$.
    \item\label{lem:complementaryConditions:2} For any $\ELL\in\Delta\cup\{\epsilon\}$, there is a ground computable simple term $t_\ELL$ such that for any answer clause $\AC$ of the form $\ansC{C[x]}{x}$, uncomputable $\cC{C'}{\ELL}\in\iteNF(\AC,t)$, and substitution $\theta$ such that $C\theta=C'$, it holds that $x\theta=t_\ELL$ or $x$ is a variable not in $C$.
    \item\label{lem:complementaryConditions:2'} If $\emptylist\notin\Delta$, then for any answer clause $\AC$ of the form $\ansC{C[x]}{x}$, uncomputable $\cCs{C'}\in\iteNF(\AC,t)$, and substitution $\theta$ such that $C\theta=C'$, it holds that $x$ is a variable not in $C$.
    \item\label{lem:complementaryConditions:3} For any $\ELL,\KAY\in\Delta$ such that $\ELL\neq\KAY$, $\ELL$ and $\KAY$ contain a complementary literal.
\end{enumerate}
\end{lemma}
\begin{proof}
By induction on the structure of $t$.

\begin{proofcase}{1}{$t$ is of the form $\iteF(L,s,s')$.}\end{proofcase}

\noindent
By the induction hypothesis, there are clause sets $S$ and $S'$ for program terms $s$ and $s'$, respectively, that satisfy the four properties. We now prove that the following computable ground clause set $\Delta$ satisfies the properties for $t$:
$$\{\lapp{\hat{L}}{\ELL}\mid \ELL\in S\}\cup\{\lapp{L}{\ELL}\mid \ELL\in S'\}$$
\noindent
Towards \Cref{lem:complementaryConditions:1}, let $\AC$ be an answer clause of the form $\ansC{C[x]}{x}$. Suppose that $C$ does not contain $x$. Then, $\iteNF(\AC,t)$ is $\{\cCs{C}\}$ and the property is satisfied.

Otherwise, we have by induction hypothesis that any clause in $\iteNF(\AC,s)$ (resp. $\iteNF(\AC,s')$) is of the form $\cC{C}{\ELL}$ where $\ELL\in S\cup\{\emptylist\}$ (resp. $\ELL\in S'\cup\{\emptylist\}$). By construction, any clause in $\iteF(\AC,t)$ is either of the form $\cC{D}{\lapp{\hat{L}}{\ELL}}$ where $\cC{D}{\ELL}\in\iteNF(\AC,s)$ or $\cC{D}{\lapp{L}{\KAY}}$ where $\cC{D}{\KAY}\in\iteNF(\AC,s')$. Both $\lapp{\hat{L}}{\ELL}$ and $\lapp{L}{\KAY}$ are in $\Delta$, proving the first property.\medskip

\noindent
Towards \Cref{lem:complementaryConditions:2}, let $\AC$ be an answer clause of the form $\ansC{C[x]}{x}$ and $\theta$ a substitution. Suppose that $C$ does not contain $x$. Then, $\iteNF(\AC,t)$ is the set $\{\cCs{C}\}$ and the property trivially holds due to $x$ being a variable not in $C$.

Otherwise, by construction, any clause in $\iteNF(\AC,t)$ is one of the following forms:
\begin{enumerate}
    \item $\cC{D}{\lapp{\hat{L}}{\ELL}}$ where $\cC{D}{\ELL}\in\iteNF(\AC,s)$: by induction hypothesis there is a ground computable term $t_\ELL$ for $\cC{D}{\ELL}$ such that $x\theta=t_\ELL$ or $x$ is a variable not in $C$ if $C[x]\theta=D$. It follows trivially for $\cC{D}{\lapp{\hat{L}}{\ELL}}$, too and we set $t_{\hat{L}\lor\ELL} = t_\ELL$.
    \item $\cC{D}{\lapp{L}{\ELL}}$ where $\cC{D}{\ELL}\in\iteNF(\AC,s')$: by induction hypothesis there is a ground computable term $t_\ELL$ for $\cC{D}{\ELL}$ such that $x\theta=t_\ELL$ or $x$ is a variable not in $C$ if $C[x]\theta=D$. It follows trivially for $\cC{D}{\lapp{L}{\ELL}}$, too and we set $t_{L\lor\ELL} = t_\ELL$.
\end{enumerate}

\noindent
Towards \Cref{lem:complementaryConditions:2'}, let $\AC$ be an answer clause of the form $\ansC{C[x]}{x}$ and $\theta$ a substitution. Suppose that $C$ does not contain $x$. Then, $\iteNF(\AC,t)$ is the set $\{\cCs{C}\}$ and the property trivially holds due to $x$ being a variable not in $C$.

Otherwise, by construction, any clause in $\iteNF(\AC,t)$ is either of the form $\cC{D}{\lapp{\hat{L}}{\ELL}}$ where $\cC{D}{\ELL}\in\iteNF(\AC,s)$ or of the form $\cC{D}{\lapp{L}{\ELL}}$ where $\cC{D}{\ELL}\in\iteNF(\AC,s')$. These clauses are not of the form $\cCs{C}$ so the property trivially holds for them.\medskip

\noindent
Towards \Cref{lem:complementaryConditions:3}, any pair of clauses in $\iteNF(\AC,s)$, respectively $\iteNF(\AC,s')$ with different conditions have complementary literals by assumption, and this is still true after adding $\hat{L}$, respectively $L$ to these clauses. All other pairs (between the two clause sets) have the complementary $L$ and $\hat{L}$ in their conditions.

\begin{proofcase}{2}{Otherwise.}\end{proofcase}

\noindent
The set $\Delta$ is simply $\{\emptylist\}$. Let $\ansC{C[x]}{x}$ be an answer clause and $\theta$ a substitution. We have that $\iteNF(\ansC{C[x]}{x},t)=\{\cCs{C[t]}\}$. By \Cref{def:ite-normal-form}, thus \Cref{lem:complementaryConditions:1} is satisfied. For \Cref{lem:complementaryConditions:2}, either $\theta=\{x\mapsto t\}$ and the term in question $t_{\emptylist}$ has to be $t$ or $C[x]\theta\neq C[t]$. \Cref{lem:complementaryConditions:2'} is trivially satisfied as $\emptylist\in\Delta$. \Cref{lem:complementaryConditions:3} is also trivially satisfied, since there are no two different elements in the set $\Delta$.
\QED
\end{proof}

\uncomputableConditions*
\begin{proof}
Let $\Delta$ be the set of ground literals obtained by applying \Cref{lem:complementaryConditions} for the term $t$.

\noindent
Towards \Cref{lemma:uncomputable-conditions:1}, let $\CC\in\Gr(S,t)$. We prove the property by induction on the derivation length of $\CC$.

\begin{proofcase}{1}{$\CC$ has a derivation length 0.}\end{proofcase}

\noindent
By \Cref{def:uncomp-grounding:base} of \Cref{def:uncomp-grounding}, $\CC$ is in $\iteNF(\AC,t)$ for some $\AC$ of the form $\ansC{C[x]}{x}$ and the claim holds due to \Cref{lem:complementaryConditions}.

\begin{proofcase}{2}{$\CC$ has non-zero derivation length.}\end{proofcase}

\noindent
By \Cref{def:uncomp-grounding:step} of \Cref{def:uncomp-grounding}, there are clauses $\cC{C_1}{\KAY_1},\ldots,\cC{C_n}{\KAY_n}$ from which $\CC$ is derived, all with smaller derivation length, so we have that \Cref{lemma:uncomputable-conditions:1} holds for all these clauses, i.e. for all $1\le i\le n$, we have $\KAY_i$ is empty or $\KAY_i\in\Delta$. From \Cref{def:uncomp-grounding:step} of \Cref{def:uncomp-grounding}, either $\KAY_i$ is empty for all $1\le i\le n$, and we have that $\CC$ is of the form $\cCs{C}$, or there is some $1\le i\le m$ such that $\CC$ is of the form $\cC{C}{\ELL}$ for $\ELL=\KAY_i$. In both cases we have $\ELL\in\Delta\cup\{\emptylist\}$. This proves the claim.\medskip

\noindent
Towards \Cref{lemma:uncomputable-conditions:2}, let $\ELL\in\Delta\cup\{\emptylist\}$ and $t_\ELL$ be the ground computable term obtained from \Cref{lem:complementaryConditions:2} of \Cref{lem:complementaryConditions}. We show that this $t_\ELL$ satisfies the property for $\ELL$. Let $\CC$ be any clause from $\Gr(S,t)$ of the form $\cC{C'}{\ELL}$. The proof is again by induction on the derivation length of $\CC$.

\begin{proofcase}{1}{$\CC$ has a derivation length 0.}\end{proofcase}

\noindent
By \Cref{def:uncomp-grounding:base} of \Cref{def:uncomp-grounding}, $\CC$ is in $\iteNF(\AC,t)$, and the claim holds for all answer clauses $\ansC{C[x]}{x}\in S$ and substitutions $\theta$ due to \Cref{lem:complementaryConditions}.

\begin{proofcase}{2}{$\CC$ has non-zero derivation length.}\end{proofcase}

\noindent
By \Cref{def:uncomp-grounding:step} of \Cref{def:uncomp-grounding}, there is an inference $\ansC{C_1}{p_1},\ldots,\ansC{C_n}{p_n}\vdash \ansC{C}{p}\sigma$, clauses $\cC{C_1'}{\KAY_1}$, $\ldots$, $\cC{C_n'}{\KAY_n}$, and a substitution $\theta$ such that $\sigma\theta=\theta$, $C\theta=\CC$, and by assumption for all $1\le i\le n$, we have $\KAY_i\in\Delta\cup\{\emptylist\}$, and by the induction hypothesis, the property holds for clause $\cC{C_i'}{\KAY_i}$, answer clause $\ansC{C_i}{p_i}$ and $\theta$.

For all $\Sup_U$, $\EqRes$ and $\EqFac$ rule applications, there is an idempotent most general unifier $\sigma$ such that $C_i\sigma\theta=C_i'$ and $p_i\sigma=p$ for all $1\le i\le n$.

Suppose that for all $1\le i\le n$, we have that $p_i$ is a variable not in $C_i$. Then, it follows that $p$ is also a variable not in $C$. Otherwise, there is at least one $1\le i \le n$ such that $p_i\sigma\theta=t_\ELL$. Note that for all $1\leq j\leq n$, if $p_j$ is not a variable not in $C_j$, also $p_j\sigma\theta=t_\ELL$ and $p_j$ is a simple term. We get that $p_i\sigma\theta=p\theta=t_\ELL$ and $p_j\sigma$ is a simple term too, proving the claim.

\medskip

\noindent
Towards \Cref{lemma:uncomputable-conditions:2'}, let $\CC$ be any clause from $\Gr(S,t)$ of the form $\cCs{C'}$. The proof is again by induction on the derivation length of $\CC$.

\begin{proofcase}{1}{$\CC$ has a derivation length 0.}\end{proofcase}

\noindent
By \Cref{def:uncomp-grounding:base} of \Cref{def:uncomp-grounding}, $\CC$ is in $\iteNF(\AC,t)$, and the claim holds for all answer clauses $\ansC{C[x]}{x}\in S$ and substitutions $\theta$ due to \Cref{lem:complementaryConditions}.

\begin{proofcase}{2}{$\CC$ has non-zero derivation length.}\end{proofcase}

\noindent
By \Cref{def:uncomp-grounding:step} of \Cref{def:uncomp-grounding}, and by the induction hypothesis, there is an inference $\ansC{C_1}{x_1},\ldots,\ansC{C_n}{x_n}\vdash \ansC{C}{p}\sigma$, clauses $\cCs{C_1'}$, $\ldots$, $\cCs{C_n'}$, and a substitution $\theta$ such that $\sigma\theta=\theta$, $C\theta=\CC$, and $x_i$ is a variable not in $C_i$ for all $1\le i\le n$. Again, by inspecting the unifier applied in all $\Sup_U$, $\EqRes$ and $\EqFac$ rule applications, we get that $x_i\sigma=p$ for all $1\le i\le n$ and therefore $p$ must be also a variable not in $C$.

\medskip

\noindent
Finally, the clause set $\Delta$ satisfies \Cref{lemma:uncomputable-conditions:3} by \Cref{lem:complementaryConditions}.\QED

\end{proof}

\groundingSatisfied*

\begin{proof}
We show that for any c-clause $\CC\in \Gr(S,t)$, it holds that $R_{\Gr(S,t)} \models \CC$. Since $\bigcup_{\AC\in S_0} \iteNF(\AC, t)$ is equisatisfiable with $\lnot F[\bar{\sigma},t[\bar{\sigma}]]$, it then follows that $\lnot F[\bar{\sigma},t[\bar{\sigma}]]$ is satisfiable too.\footnote{The initial set $S_0$ is $\ansC{\cnf(\neg F[\bar{\alpha}, y])}{y}$.}

The proof is by induction on $\succ$ on c-clauses. By contradiction, suppose that there is a c-clause $\CC\in \Gr(S,t)$ such that $R_{\Gr(S,t)}\nmodels \CC$. Since $\succ$ is well-founded, there is a minimal such $\CC$. Since all c-clauses in $\Gr(S,t)\setminus \CGr(S)$ are uncomputable and greater w.r.t. $\succ$ than any (necessarily computable) c-clause in $\CGr(S)$, by~\Cref{lem:Rproperties} and~\Cref{lemma:computable-clauses-completeness} we have that $\CC$ is satisfied if $\CC\in\CGr(S)$. Hence, $\CC$ is uncomputable. We distinguish the following cases.


\begin{proofcase}{1}{$\CC$ is of the form $\cC{C'}{\ELL}$ and there is an answer clause $\ansC{C}{p}$ in $S$ and a substitution $\theta$ s.t. $C\theta=C'$, and there is a variable $x$ in $C$ such that $x\theta$ is reducible in $R_{\Gr(S,t)}$.}\end{proofcase}

\noindent
Note that $p$ cannot contain $x$ by~\Cref{def:uncomp-grounding}. Let $\theta'$ be the substitution $\theta$, except that it maps $x$ to the normal form of $x\theta$ w.r.t. $R_{\Gr(S,t)}$. Let $\CD$ be the c-clause $\cC{C'\theta'}{\ELL}$. Then, since $x\theta\succ x\theta'$ and $C$ contains $x$, we have $\CC\succ \CD$. Also, $x\theta'$ is the normal form of $x\theta$, so $R_{\Gr(S,t)}\nmodels \CD$ too. Note that $\CD$ is in $\Gr(S,t)$ by \Cref{def:uncomp-grounding}. Thus, we have found a c-clause in $\Gr(S,t)$ smaller than $\CC$.

\begin{proofcase}{2}{$\CC$ is of the form $\cC{\underline{s'\neqs t'}\lor C'}{\ELL}$ such that $s'\neqs t'$ is maximal in $\CC$.}\end{proofcase}

\begin{proofsubcase}{2.1}{$s'=t'$.}\end{proofsubcase}

\noindent
There is a clause $\ansC{\underline{s\neqs t}\lor C}{p}$ in $S$, a ground substitution $\theta$ such that $(s\neqs t\lor C)\theta=s'\neqs t'\lor C'$. Consider the following inference:
\begin{prooftree}
    \AxiomC{$\ansC{\underline{s\neqs t}\lor C}{p}$}
    \LeftLabel{($\EqRes$)}
    \UnaryInfC{$(\ansC{C}{p})\sigma$}
\end{prooftree}
where $\sigma=\mgu(s,t)$. Let $\CD$ be $\cC{C'}{\ELL}$. By $s\theta=t\theta=s'$, we have that $\theta$ is a unifier of $s$ and $t$. By \Cref{lemma:uncomputable-conditions}, we get that $p\theta$ is computable or $p$ is a variable and not in $s\neqs t\lor C$, and therefore $p\sigma$ is also computable. Therefore, condition 2 of the $\EqRes$ rule is satisfied and $\ansC{C}{p}\sigma\in S$ and $\CD\in \Gr(S,t)$. Due to $\{\{s',s',s',s'\}\}\uplus C'\succbag C'$, we also have $\CC \succ \CD$. Also, $R_{\Gr(S,t)}\nmodels C'$, so $R_{\Gr(S,t)} \nmodels \CD$. Thus, we have found a false clause smaller than $\CC$, contradiction.

\begin{proofsubcase}{2.2}{W.l.o.g. $s'\succ t'$ and $s'$ is reducible by some $l'\to r'\in R_{\Gr(S,t)}$ where $l'$ is not computable.}\end{proofsubcase}

\noindent
Then $s'$ is of the form $s'[l']$ and there is a rule $l'\to r'$ in $R_{\Gr(S,t)}$ produced by some uncomputable clause $\cC{l'\eqs r'\lor D'}{\KAY}$ in $\Gr(S,t)$. Let $\Delta$ be the set from \Cref{lemma:uncomputable-conditions}. By \Cref{lemma:uncomputable-conditions:1} of \Cref{lemma:uncomputable-conditions}, we have that $\ELL,\KAY\in\Delta\cup\{\emptylist\}$. If they are in $\Delta$ and different, then by \Cref{lemma:uncomputable-conditions:3} of \Cref{lemma:uncomputable-conditions}, we have that $\ELL$ and $\KAY$ contain complementary conditions, that is, there is a literal $L$ such that $L\in\ELL$ and $\hat{L}\in\KAY$. By \Cref{productive:false2} of \Cref{def:superposition-model}, $R_{\Gr(S,t)}\nmodels\hat{L}$, and thus $R_{\Gr(S,t)}\models L$, so $\CC$ cannot be false, contradiction.

Otherwise $\ansC{\underline{s[k]\neqs t}\lor C}{p}$ and $\ansC{\underline{l\eqs r}\lor D}{q}$ in $S$ such that $k$ is not a variable (see Case 1), and there is a ground substitution $\theta$ such that $l\theta=k\theta=l'$, $(s[k]\neqs t\lor C)\theta=s'[l']\neqs t'\lor C'$ and $(l\eqs r\lor D)\theta=l'\eqs r'\lor D'$. If $\ELL,\KAY\in \Delta$, then they must be the same, and by \Cref{lemma:uncomputable-conditions:2} of \Cref{lemma:uncomputable-conditions}, we get that there is a term $t_\ELL$ ($=t_\KAY$) such that (i) $q\theta=t_\ELL$ or $q$ is a variable not in $l\eqs r\lor D$ and (ii) $p\theta=t_\ELL$ or $p$ is a variable not in $s[k]\neqs t\lor C$. Otherwise $\ELL$ or $\KAY$ is empty, and then by \Cref{lemma:uncomputable-conditions:2'} of \Cref{lemma:uncomputable-conditions}, we get that $p$ (resp. $q$) is a variable that is not in $s[k]\neqs t\lor C$ (resp. $l\eqs r\lor D$). In all of these cases, we get that the unifier $\sigma=\mgu((l,p),(k,q))$ exists. Consider the following inference:
\begin{prooftree}
    \AxiomC{$\ansC{\underline{l\eqs r}\lor D}{q}$}
    \AxiomC{$\ansC{\underline{s[k]\neqs t}\lor C}{p}$}
    \LeftLabel{($\Sup_U$)}
    \BinaryInfC{$\ansC{s[r]\neqs t\lor C\lor D}{p}\sigma$}
\end{prooftree}
Since $l'\eqs r'\lor D'$ is productive, \Cref{productive:ordered} from \Cref{def:superposition-model} and the assumption $s'\succ t'$ imply that $r\sigma\nsucceq l\sigma$ and $t\sigma\nsucceq s[k]\sigma$ (condition 3 of the $\Sup_U$ rule). By the above reasoning, we also have that $p\theta$ is either $t_\ELL$, $t_\KAY$ or a variable, therefore $p\sigma$ is computable and condition 4 of the $\Sup_U$ rule is satisfied. Then we have that $\ansC{s[r]\neqs t\lor C\lor D}{p}\sigma\in S$. Let $\CD$ be the c-clause $\cC{s'[r']\neqs t'\lor C'\lor D'}{\ELL'}$ where $\ELL'$ is $\ELL$ if $\KAY$ is empty, otherwise $\KAY$. By \Cref{def:uncomp-grounding:step} of \Cref{def:uncomp-grounding}, we have $\CD\in \Gr(S,t)$. Similarly to Subcase 3.2 of \Cref{lemma:computable-clauses-completeness}, we get that $\CC\succ\CD$ and $s'[r']\neqs t'\lor C'\lor D'$ is false. Note that $\ELL'$ is either empty or false because $\ELL$ is false (by \Cref{productive:false2} of \Cref{def:superposition-model}). Thus, $\CD$ is a false clause in $\Gr(S,t)$ smaller than $\CC$, contradiction. 

\begin{proofsubcase}{2.3}{W.l.o.g. $s'\succ t'$ and $s'$ is reducible by some $l'\to r'\in R_{\Gr(S,t)}$ where $l'$ is computable.}\end{proofsubcase}

\noindent
Then $s'$ is of the form $s'[l']$ and there is a rule $l'\to r'$ in $R_{\Gr(S,t)}$ produced by some computable clause $\cCs{l'\eqs r'\lor C'}$. Suppose that $s'$ is computable. Then, since $s'\succ t'$, $t'$ must also be computable. But then, since $\CC$ is uncomputable, there must be an uncomputable literal in $C'$ and this literal is strictly greater than $s'\neqs t'$ w.r.t. $\succ$, which contradicts the assumption that $s'\neqs t'$ is maximal. Therefore, $s'$ must be uncomputable. This also means that $l'$ is a strict subterm of $s'$, otherwise $s'$ would be computable.

There is $\ansC{\underline{s[l]\neqs t}\lor C}{p}$ in $S$ such that $l$ is not a variable, a ground substitution $\theta$ such that $l\theta=l'$ and $(s[l]\neqs t\lor C)\theta=s'[l']\neqs t'\lor C'$. Consider the following inference where the variable $x$ is fresh and $k$ is some superterm of $l$ (possibly $l$ itself):
\begin{prooftree}
\AxiomC{$\ansC{s[k]\neqs t\lor C}{p}$}
\LeftLabel{($\Abs$)}
\UnaryInfC{$\ansC{s[x]\neqs t\lor x\neqs k\lor C}{p}$}
\end{prooftree}
We have $s'\succ t'$, so $s\npreceq t$ satisfying condition 1 of the $\Abs$ rule. Condition 2 of the $\Abs$ rule is satisfied for $l$ and $\theta$. W.l.o.g. we can assume that $k$ is chosen such that condition 2 is still satisfied, and such that no superterm of $k$ satisfies condition 2, satisfying condition 3 of the rule. But this means that the $\Abs$ rule was not applied to this clause, and thus $S$ is not abstracted, contradiction.

\begin{proofsubcase}{2.4}{W.l.o.g. $s'\succ t'$ and $s'$ is irreducible.}\end{proofsubcase}

\noindent
Then, since $s'\succ t'$ and $s'$ is in normal form, $R_{\Gr(S,t)}\models s'\neqs t'$ and $\CC$ is satisfied, contradiction.

\begin{proofcase}{3}{$\CC$ is of the form $\cC{\underline{s'\eqs t'}\lor C'}{\ELL}$ where $s'\eqs t'$ is maximal in $s'\eqs t'\lor C'$.}\end{proofcase}

\noindent
W.l.o.g. $s'\succ t'$.

\begin{proofsubcase}{3.1}{$s'$ is reducible by some $l'\to r'\in R_{\Gr(S,t)}$ where $l'$ is uncomputable.}\end{proofsubcase}

\noindent
This is similar to Subcase 2.2. Then $s'$ is of the form $s'[l']$ and $l'\to r'$ is produced by some c-clause $\cC{l'\eqs r'\lor D'}{\KAY}$. There are clauses $\ansC{\underline{s[k]\eqs t}\lor C}{p}$ and $\ansC{\underline{l\eqs r}\lor D}{q}$ in $S$, such that $k$ is not a variable, a ground substitution $\theta$ such that $l\theta=k\theta=l'$, $(s[k]\eqs t\lor C)\theta=s'[l']\eqs t'\lor C'$ and $(l\eqs r\lor D)\theta=l'\eqs r'\lor D$. Similarly to Subcase 2.2, we assert that either $R_{\Gr(S,t)}\models \CC$ reaching a contradiction, or that the unifier $\sigma=\mgu((l,p),(k,q))$ exists. Consider the following inference:
\begin{prooftree}
    \AxiomC{$\ansC{\underline{l\eqs r}\lor D}{q}$}
    \AxiomC{$\ansC{\underline{s[k]\eqs t}\lor C}{p}$}
    \LeftLabel{($\Sup_U$)}
    \BinaryInfC{$\ansC{s[r]\lor C\lor D}{p}\sigma$}
\end{prooftree}
Let $\CD$ be the c-clause $\cC{s'[r']\eqs t'\lor C'\lor D'}{\ELL'}$ where $\ELL'$ is $\ELL$ if $\KAY$ is empty, otherwise $\KAY$. We check the conditions the $\Sup_U$ rule similarly as in Subcase 2.2. Therefore, $\ansC{s[r]\lor C\lor D}{p}\sigma\in S$ and $\CD\in \Gr(S,t)$. We also have the following
$$\{\{s'[l'],t'\}\}\uplus C'\succbag \{\{s'[r'],t'\}\}\uplus C'\uplus D'$$
and hence $\CC \succ \CD$. Also, $R_{\Gr(S,t)}\nmodels \CD$: (1) by assumption, $C'$ and $D'$ are false in $R_{\Gr(S,t)}$, and (2) $s'[l']$ have distinct normal forms with $t'$, so $s'[r']$ and $t'$ must have distinct normal forms too (otherwise $\CC$ would be true in $R_{\Gr(S,t)}$ by confluence), and (3) $\ELL'$ is either empty or false because $\ELL$ is empty or false by \Cref{productive:false2} of \Cref{def:superposition-model}. Thus, we have found a false clause in $S$ smaller than $\CC$.

\begin{proofsubcase}{3.2}{W.l.o.g. $s'\succ t'$ and $s'$ is reducible by some $l'\to r'\in R_{\Gr(S,t)}$ where $l'$ is computable.}\end{proofsubcase}

\noindent
Analogous to Subcase 2.3.

\begin{proofsubcase}{3.3}{$s'$ is irreducible in $R_{\Gr(S,t)}$.}\end{proofsubcase}

\noindent
Then, $C'$ is of the form $s'\eqs r'\lor D'$ such that $t'\succeq r'$ and $s'\eqs r'$ is true in $R_{\CGr(S)}^{\prec \CC}\cup\{s'\to t'\}$ (i.e. \Cref{productive:false3} of \Cref{def:superposition-model} is violated). Note that this also covers the case when \Cref{productive:maximal} of \Cref{def:superposition-model} is violated, as then $t'=r'$ and $\CC$ is of the form $\cC{\underline{s'\eqs t'}\lor \underline{s'\eqs t'}\lor D'}{\ELL}$.

Then there is an answer clause $\ansC{\underline{s\eqs t}\lor l\eqs r\lor D}{p}$ in $S$, a ground substitution $\theta$ such that $(s\eqs t\lor l\eqs r\lor D)\theta=s'\eqs t'\lor s'\eqs r'\lor D'$. Consider the following inference:
\begin{prooftree}
    \AxiomC{$\ansC{\underline{s\eqs t}\lor l\eqs r\lor D}{p}$}
    \LeftLabel{($\EqFac$)}
    \UnaryInfC{$\ansC{s\eqs t\lor t\neqs r\lor D}{p}\sigma$}
\end{prooftree}
where $\sigma=\mgu(s,l)$. Let $\CD$ be the c-clause $s'\eqs t'\lor t'\neqs r'\lor D'$. Since $s'\succ t'$ and $t'\succeq r'$ by assumption, condition 2 of the $\EqFac$ rule holds. By \Cref{lemma:uncomputable-conditions}, we get that $p\theta$ is computable or $p$ is a variable and not in $s\neqs t\lor C$, and therefore $p\sigma$ is also computable, satisfying condition 3 of $\EqFac$. Then, $\ansC{s\eqs t\lor t\neqs r\lor D}{p}\sigma\in S$ and $\CD\in \Gr(S,t)$. Similarly to Subcase 4.2 of \Cref{lemma:computable-clauses-completeness}, we get that $\CC \succ \CD$ and $R_{\Gr(S,t)}\nmodels (s'\eqs t'\lor t'\neqs r'\lor D')$. We also have that $R_{\Gr(S,t)}\nmodels \ELL$, hence $R_{\Gr(S,t)}\nmodels\CD$. Thus, we have found a false clause in $\Gr(S,t)$ smaller than $\CC$.\QED

\end{proof}

}
\end{document}